\documentclass[a4paper,onecolumn,11pt,accepted=2024-03-24]{quantumarticle}
\pdfoutput=1
\usepackage{makecell}
\usepackage{array}
\newcolumntype{C}[1]{>{\centering\arraybackslash}m{#1}}
\usepackage{booktabs,adjustbox}
\usepackage{float}
\usepackage{amsmath}
\usepackage[utf8]{inputenc}  
\usepackage[T1]{fontenc}     
\usepackage[table]{xcolor}
\usepackage{color,soul}
\usepackage[british]{babel}  

\usepackage[colorlinks=true, citecolor=blue, urlcolor=blue]{hyperref}  
\usepackage{graphicx} 
\usepackage[babel]{microtype}  
\usepackage{amsmath,amssymb,amsthm,bm,amsfonts,mathrsfs,bbm} 

\usepackage{xspace}  
\usepackage{pgf,tikz}
\usepackage{xcolor}
\usepackage{multirow}
\usepackage{array}
\usepackage{bigstrut}
\usepackage{braket}
\usepackage{color}
\usepackage[numbers]{natbib}
\usepackage{multirow}
\usepackage{mathtools}
\usepackage{float}
\usepackage{blkarray}
\usepackage[caption = false]{subfig}
\usepackage{xcolor,colortbl}
\usepackage{color}
\usepackage{rotating}
\usepackage{tikz}
\usetikzlibrary{quantikz}
\usepackage{tabularx}
\usepackage{subfig}
\newcommand{\Tr}{\operatorname{Tr}}

\newcommand{\be}{\begin{equation}}
\newcommand{\ee}{\end{equation}}
\newcommand{\ba}{\begin{eqnarray}}
\newcommand{\ea}{\end{eqnarray}}
\newcommand{\ketbra}[2]{|#1\rangle \langle #2|}
\newcommand{\tr}{\operatorname{Tr}}

\newtheorem{theorem}{Theorem}

\newtheorem{definition}{Definition}
\newtheorem{proposition}{Proposition}

\newtheorem{remark}{Remark}

\def\>{\rangle}
\def\<{\langle}

\usepackage{centernot}
\usepackage{subfig}

\begin{document}

\title{Classical analogue of quantum superdense coding and communication advantage of a single quantum system}

\author{Ram Krishna Patra}
\affiliation{Department of Physics of Complex Systems, S.N. Bose National Center for Basic Sciences, Block JD, Sector III, Salt Lake, Kolkata 700106, India.}

\author{Sahil Gopalkrishna Naik}
\affiliation{Department of Physics of Complex Systems, S.N. Bose National Center for Basic Sciences, Block JD, Sector III, Salt Lake, Kolkata 700106, India.}

\author{Edwin Peter Lobo}
\affiliation{Laboratoire d’Information Quantique, Université libre de Bruxelles (ULB), Av. F. D. Roosevelt 50, 1050 Bruxelles, Belgium}

\author{Samrat Sen}
\affiliation{Department of Physics of Complex Systems, S.N. Bose National Center for Basic Sciences, Block JD, Sector III, Salt Lake, Kolkata 700106, India.}

\author{Tamal Guha}
\affiliation{Department of Computer Science, The University of Hong Kong, Pokfulam Road, Hong Kong.}

\author{Some Sankar Bhattacharya}
\affiliation{International Centre for Theory of Quantum Technologies, University of Gdansk, Wita Stwosza 63, 80-308 Gdansk, Poland.}

\author{Mir Alimuddin}
\affiliation{Department of Physics of Complex Systems, S.N. Bose National Center for Basic Sciences, Block JD, Sector III, Salt Lake, Kolkata 700106, India.}

\author{Manik Banik}
\affiliation{Department of Physics of Complex Systems, S.N. Bose National Center for Basic Sciences, Block JD, Sector III, Salt Lake, Kolkata 700106, India.}

\begin{abstract}
We analyze utility of communication channels in absence of any short of quantum or classical correlation shared between the sender and the receiver. To this aim, we propose a class of two-party communication games, and show that the games cannot be won given a noiseless $1$-bit classical channel from the sender to the receiver. Interestingly, the goal can be perfectly achieved if the channel is assisted with  classical shared randomness. This resembles an advantage similar to the quantum superdense coding phenomenon where pre-shared entanglement can enhance the communication utility of a perfect quantum communication line. Quite surprisingly, we show that a qubit communication without any assistance of classical shared randomness can achieve the goal, and hence establishes a novel quantum advantage in the simplest communication scenario. In pursuit of a deeper origin of this advantage, we show that an advantageous quantum strategy must invoke quantum interference both at the encoding step by the sender and at the decoding step by the receiver. We also study communication utility of a class of non-classical toy systems described by symmetric polygonal state spaces. We come up with communication tasks that can be achieved neither with $1$-bit of classical communication nor by communicating a polygon system, whereas $1$-qubit communication yields a perfect strategy, establishing quantum advantage over them. To this end, we show that the quantum advantages are robust against imperfect encodings-decodings, making the protocols implementable with presently available quantum technologies.             
\end{abstract}



\maketitle
\section{Introduction}
Present day digital era crucially depends on reliable transfer of information among distant locations through several advanced means, such as cellular, internet, and satellite communications. Foundation of modern communication theory was established in the seminal work of Claude E. Shannon, who modeled the physical devices to store and transfer the information as classical objects \cite{Shannon48}. The advent of quantum information theory \cite{Nielsen10}, also recognized as 'the second quantum revolution' \cite{Dowling03}, identifies novel uses of non-classical features of quantum systems to devise exotic information and communication protocols that are advantageous over their classical counterparts and in some cases impossible with classical resources \cite{Bennett92,Bennett93,Bennett00,Kimble08,Dale15,Zhang17,Boes18,Rosset18,Ebler18,Korzekwa21,Chiribella21,Bhattacharya21,Koudia21}. Over the last few decades, hands-on uses of quantum resources have been reported in a number of crucial experiments (see \cite{Bouwmeester97,Gisin02,Georgescu14,Degen17} and references therein), and with every passing day, quantum technology is stepping out of merely academic interests to more practical uses \cite{Yin17,Valivarthi20,Xu20}, and thus motivates finding out more-and-more novel communication protocols wherein quantum resources exhibit advantages over their classical counterparts.

The simplest communication scenario involves two distant parties -- a sender (Alice) and a receiver (Bob) -- where Alice aims to transmit some message to Bob by sending some physical systems. The pioneering `quantum superdense coding' protocol establishes a quantum advantage by showing that quantum entanglement, preshared between sender and receiver, can double the classical communication capacity of a perfect qubit channel \cite{Bennett92}. This is quite striking,\footnote{See the illuminating remark by the reviewer of `quantum superdense coding' paper \cite{Mermin04}.} as, in such a scenario, quantum entanglement on its own has no communication utility, and according to the fundamental no-go theorem of Holevo \cite{Holevo73}, the communication capacity of a perfect qubit channel alone is no more than the capacity of a perfect one-bit classical channel. Recently, the no-go implication of Holevo has been strengthened further by Frenkel \& Weiner while evaluating the communication utility of a quantum system in absence of preshared entanglement \cite{Frenkel15}. While in Holevo's theorem, communication utility is measured in terms of the mutual information between the random variable the sender intends to send and the random variable the receiver obtains after the channel action, Frenkel \& Weiner quantify a channel's utility through a generic reward function rather than only mutual information and still establish that the communication utility of an $n$-level quantum system is the same as that of an $n$-level classical system.

The contribution of the present work starts with the following pivotal observation: while quantum entanglement and classical shared randomness both have zero communication utility on their own, the no-go results in \cite{Holevo73,Frenkel15} consider the former to be a costly resource, and hence not allowed to be preshared between the parties, whereas the latter is freely available between them. This assumption is supported by the general practice within the study of Bell's nonlocality \cite{Bell64,Bell66,Brunner14(0)}, wherein classical shared randomness is considered to be free as they result only in 'local' correlations, while quantum entanglement can result in puzzling 'nonlocal' correlations \cite{Wolfe20,Schmid20,Rosset20}. However, it is important to note that to create classical shared randomness between two distant parties, classical communication is a necessary resource, and hence in this work we will consider it to be a costly resource. A number of works in other branches of research  \cite{Aumann87,Babai97,Canonne17} as well as in quantum information theory  \cite{Toner03,Bowles15,Llobet15,Guha21} already exist where nontrivial utilities of classical shared randomness have been pointed out. 

\subsection*{Major findings of the present work}
We investigate the simplest communication scenario involving only two parties. Like Frenkel \& Weiner \cite{Frenkel15}, we consider general reward function(s) to quantify the communication utility of the physical system transferred from the sender to the receiver, but unlike them, we assume that no classical shared randomness is preshared between the parties. Our main findings are the following: 
\begin{itemize}
\item[(i)] First, we propose a class of two-party games denoted as $\mathbb{H}^n(\gamma_1,\cdots,\gamma_n)$, where a game is specified by the parameters $\gamma_i\ge0$ with $\sum_i\gamma_i=1$. The special case where all the $\gamma_i$'s are equal will be designated with the simplified notation $\mathbb{H}^n(1/n)$. We show {in Section} \ref{1bit1SR} that only a measure zero subclass of these games can be won perfectly if $1$-bit of classical communication is allowed from the sender (Alice) to the receiver (Bob). However, all such games can be won deterministically if the classical communication line is assisted by additional shared randomness. Following the quantum superdense coding phenomenon, wherein the single-shot communication utility of a perfect quantum channel is enhanced by the assistance of quantum correlation (entanglement), we call this the `classical superdense coding' phenomenon.

\item[(ii)] We then show in section \ref{quantum strategy} that a class of aforesaid games are perfectly winnable with $1$-qubit communication from Alice to Bob that are otherwise not winnable with $1$-bit of classical communication. This establishes a novel communication utility of quantum system over its classical counterpart. To identify the origin of quantum advantage we further prove two no-go results. We show  the simultaneous use of quantum interference at the encoding step by Alice and at the decoding step by Bob is necessary for this particular advantage.

\item[(iii)] We analyze the communication utility of a class of toy models, known as the polygon models \cite{Janotta11}. Like two-level classical (bit) and quantum (qubit) systems, all these models allow at most two perfectly distinguishable states. However, with the help of our proposed game, we show in Section \ref{polygon} that the communication utility of all the even-gons is strictly greater than the two-level classical system. 

\item[(iv)] We then propose a stricter version of the aforesaid game which we denote as $\mathbb{H}^n(1/(n-1))$. In  section \ref{msrmu} we prove that the $\mathbb{H}^4(1/3)$ game cannot be won perfectly with $1$-bit of classical communication even when assisted with $1$-bit of shared randomness. More shared randomness is required to achieve the goal. This provides further justification to consider shared randomness as a costly resource.  

\item[(v)] Interestingly, we then show that $\mathbb{H}^4(1/3)$ game can be won with 1-qubit communication even without any assistance of classical shared randomness. Furthermore, no polygon system provides a perfect strategy for this game, making the advantage quantum sensitive[\ref{strictq},\ref{qbtp}].

\item[(vi)] We analyze the robustness of the quantum protocols to noise so that this newly obtained communication advantage of the quantum system can be tested with imperfect experimental devices[\ref{noise}]. 
\item[(vii)] Finally, we also explore the power of shared randomness in enhancing a different utility of a 1-bit classical channel. We consider the worst-case success probability of a guessing game as the utility and show that the utility of 1-bit classical channel in the  is $\frac{1}{2}$ in the absence of preshared correlation, whereas with the assistance of $\log 3$ bit of shared randomness this can be enhanced to $\frac{2}{3}$ [\ref{worst-case success}].
 
\end{itemize}
While presenting our results, we discuss several implications and compare them with the already existing results that are relevant. Before moving to the technical proofs of our results, in the following section, we discuss several possible communication scenarios involving two parties which will be helpful to appreciate our contributions in this work.   

\section{Preliminaries: The set up}
While studying two-party communication tasks following three distinct scenarios can be considered:
\begin{itemize}
\item[{\bf A.}] Holevo - Frenkel \& Weiner (H-FW) scenario,
\item[{\bf B.}] Wiener - Ambainis, Nayak, Ta-Shma, \& Vazirani (W-ANTV) scenario, and
\item[{\bf C.}] Bell scenario. 
\end{itemize}
Alice and Bob may be given access to different types of resources which can be classified into two broad categories.

{\bf Type-$1$ (Common-past resources):} Alice and Bob meet in their common past before the game starts and share some correlated systems which they may use later to optimize their payoff(see Fig.\ref{fig15}). In literature, generally three different kinds of resources are considered:
\begin{itemize}
\item[C$_p$-$1$:] Classical shared randomness which we will sometimes call just shared randomness,
\item[C$_p$-$2$:] Quantum entanglement \cite{Horodecki09} which can yield puzzling `nonlocal' correlations, and
\item[C$_p$-$3$:] Beyond quantum nonlocal correlations, {\it eg.} the Popescu-Rohrlich (PR) correlation \cite{Popescu94}.  
\end{itemize}

\begin{figure}[h!]
	\centering
	\subfloat[Common past resources (Type-1)]{\includegraphics[width=0.50\linewidth]{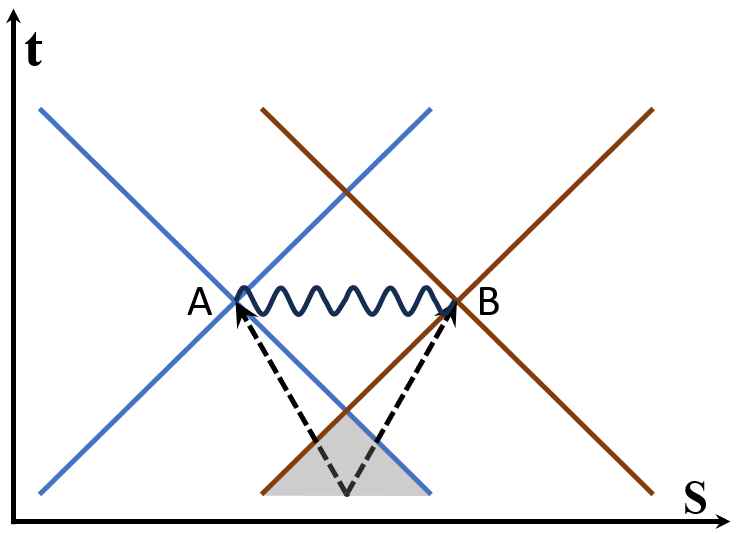}\label{fig14}}
 \hspace{1cm}
	\subfloat[Direct communication resource (Type-2)]{\includegraphics[width=0.35\linewidth]{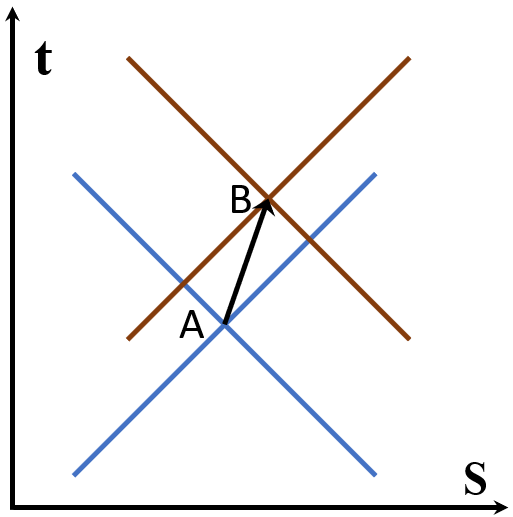}\label{fig15}}\\
 \subfloat[Combination of Direct and Common past resources]{\includegraphics[width=0.45\linewidth]{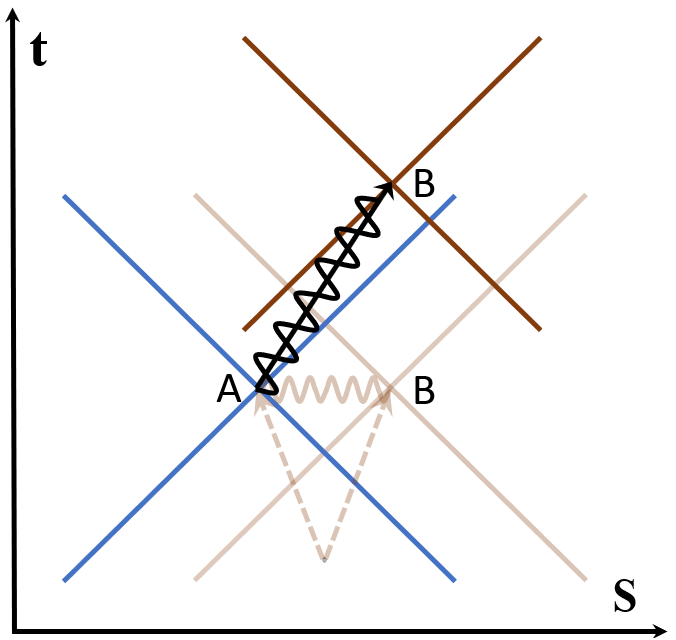}\label{fig16}}
 \caption{Causal structure of different communication resources. (a) The Shaded region is in the common past of both Alice and Bob. They can create some correlated system like shared randomness, entanglement, etc, which can be used later however their communication capacity is zero. These are known as Common past resources ({\bf Type-$1$}). (b) Alice is in the causal past of Bob, Alice sends some physical systems (e.g., classical bits, quantum bits, etc.), and Bob receives them at a later time. These are known as Direct communication resource({\bf Type-$2$}) and all of them has some nonzero communication capacity. (c) This represents the combination of the two scenarios mentioned earlier. In the beginning, they establish a shared resource ({\bf Type-$1$}) in their common past. Then Alice sends a physical system ({\bf Type-$2$}) that is received by Bob in his subsequent causal future. }\label{figcausal}
\end{figure}

{\bf Type-$2$ (Direct-communication resources):} Alice takes some physical system and encodes her message in different (allowed) states of the system(see Fig.\ref{fig15}). She then sends the system to Bob who may perform some measurements to obtain an outcome and accordingly optimize the required payoff. In this work, we will analyze three different kinds of systems\footnote{Such systems are assumed to have some common properties. One such property is the number of perfectly distinguishable states allowed in the theory, also known as the `measurement dimension' of the system \cite{Brunner14}.}: 
\begin{itemize}
\item[D$_c$-$1$:] Classical systems with state space described by a simplex, {\it eg.} a two-level classical system or a bit,
\item[D$_c$-$2$:] Quantum systems with state space $\mathcal{D}(\mathbb{C}^d)$\footnote{$\mathcal{D}\left(\mathcal{H}_S\right)$ denotes the convex-compact set of density operators acting on the Hilbert space $\mathcal{H}_S$ associated with the quantum system $S$.}, {\it eg.} a qubit when $d=2$, and
\item[D$_c$-$3$:] Some hypothetical system with a convex state space \cite{Barrett07}, {\it eg.} polygon systems \cite{Janotta11}.  
\end{itemize}
However, one can consider a communication scenario where both {\bf Type-$1$} and {\bf Type-$2$} resources are used together. The diagram (see Fig.\ref{fig16}) above depicts the causal structure of both Direct communication resources and Common past resources when used in combination.
\par
All three scenarios -- H-FW, W-ANTV, and Bell scenario -- are important to analyze the communication utilities as well as physical-hypothetical properties of the Common-Past and Direct-Communication resources. In the following, we briefly discuss  each of the scenarios to make the distinctions among them clear. 

\subsection{H-FW Scenario} This is the most basic scenario involving two-party communication. Although the study of communication using classical systems within this scenario originally started in the seminal work of Shannon \cite{Shannon48}, in the quantum domain it was introduced by Holevo (H) to understand the limitations of information transmission by a quantum channel \cite{Holevo73}. More recently, Frenkel \& Weiner (FW) have proposed a generalization of this framework \cite{Frenkel15}. Suppose, Alice and Bob are two distant parties. A Referee provides Alice some classical random variable $x\in\mathcal{X}$, while Bob needs to generate some classical random variable $b\in\mathcal{B}$ and return it back to the Referee, according to the correlation $P(b|x)$ generated by Alice and Bob refree gives them some payoff $\beta:P(\mathcal{B}|\mathcal{X})\to\mathbb{R}$  (see Fig.\ref{fig1}). The payoff can be defined on a set of correlations satisfying some constrain or on the achievability of some particular correlation.
\begin{figure}[b!]
\centering
\includegraphics[width=0.45\textwidth]{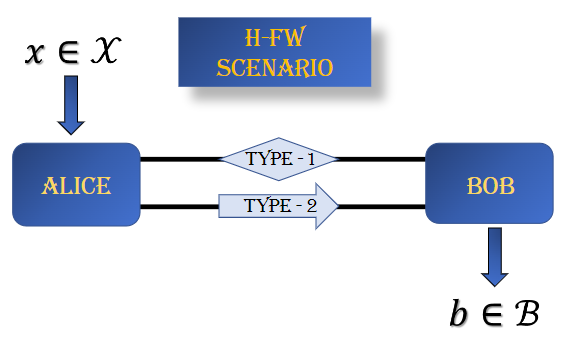}
\caption{(Color online) Holevo-Frenkel-Weiner (H-FW) communication scenario. Only Alice is given some classical input $x\in\mathcal{X}$, while only Bob has to generate a classical output $b\in\mathcal{B}$. Two different types of resources - {\bf Type-$1$} and {\bf Type-$2$} - might be utilized to optimize the targeted payoff $\beta:P(\mathcal{B}|\mathcal{X})\to\mathbb{R}$.}\label{fig1}
\end{figure}

While none of the {\bf Type-$1$} resources has communication utility  on their own in the H-FW scenario, to create a resource of the kind C$_p$-$1$ Alice and Bob must communicate to each other if they do not meet in their common past. On the other hand,  the effects of the resource of kind C$_p$-$2$ can be simulated with sufficient classical communication between Alice and Bob \cite{Hall11,Banik13}. However, the resource of the kind C$_p$-$2$, {\it i.e.} quantum entanglement, cannot be created with local operations and classical communication (LOCC) between Alice and Bob \cite{Horodecki09}. The only way they can possess it is by preparing it in their common past or if it is given by some third agent.

Although {\bf Type-$1$} resources have zero communication utility, they can enhance the communication utility of the {\bf Type-$2$} resources. The seminal example is the `quantum superdense coding' protocol, which shows that preshared entanglement between Alice and Bob can double the classical communication capacity of a perfect quantum channel \cite{Bennett92} (see \cite{Schaetz04,Barreiro08,Williams17} for experimental implementations of the protocol). Later, it has been shown that entanglement can increase the classical capacity of some noisy quantum channels by an arbitrarily large constant factor over their best-known classical capacity achievable without entanglement \cite{Thapliyal99}. Very recently, it has been shown that entanglement can also enhance the communication utility of a perfect classical channel \cite{Frenkel21}, within this H-FW scenario. Furthermore, the assistance of post-quantum correlation in this case turns out to be more useful than quantum entanglement. Here we note that the success probability of the game proposed in \cite{Frenkel21} with the resource of $1$ bit of classical communication along with some generic no signaling (NS) correlation scales linearly with the famous Clauser-Horne-Shimony-Holt (CHSH) expression \cite{Clauser69} (see Appendix \ref{appendix-a}).

In this H-FW communication scenario, the celebrated no-go result by Holevo limits the information capacity of an $n$-level quantum system by the optimal value achievable with an $n$-level classical system \cite{Holevo73}. While in Holevo's theorem the utility is calculated through an entropic quantity, namely the mutual information $I(\mathcal{X}:\mathcal{B})$ between Alice's input random variable and Bob's output random variable, Frenkel \& Weiner evaluate how successfully Alice and Bob can manage to store and recover the value of $x\in\mathcal{X}$ by requiring Bob to specify a value $b\in\mathcal{B}$ and giving a generic reward of value $f(x, b)$ to the team. They have shown that whatever the probability distribution of $x$ and the reward function $f$ are, when using a quantum $n$-level system, the maximum expected reward obtainable with the best possible team strategy is equal to that obtainable with the use of a classical $n$-level system. Furthermore, like {\it Shannon’s Noisy Channel Coding Theorem}, Holevo's theorem captures the channel's capacity of reliable transmission rate in the asymptotic limit, while with a single use of the channel things can go differently. Therefore, the result of Frenkel \& Weiner should be seen as an independent no-go theorem while evaluating the communication utility of a quantum system. 

It is important to note that Frenkel \& Weiner, while deriving their no-go result, consider classical shared randomness (C$_p$-$1$ resource of {\bf Type-$1$}) as free, {\it i.e.} Alice and Bob can have arbitrary amount of shared randomness in their possession (see also \cite{DallArno17}). We have already seen that C$_p$-$2$ and C$_p$-$3$ resources of {\bf Type-$1$} can empower communication utility of different resources of {\bf Type-$2$} \cite{Bennett92,Schaetz04,Barreiro08,Williams17,Thapliyal99,Frenkel21}. It is, therefore, a natural question to analyze the communication utility of different resources of {\bf Type-$2$} in absence of C$_p$-$1$ resource of {\bf Type-$1$}, {\it i.e.} shared randomness. It is precisely this question that the present work aims to address.

\subsection{W-ANTV Scenario}
This scenario in more involved and also distinct from the earlier scenario. In this case, the Referee provides random variables $x\in\mathcal{X}$ and $y\in\mathcal{Y}$ to Alice and Bob, respectively and only Bob needs to return an output random variable $b\in\mathcal{B}$. Accordingly, they will be given some payoff depending on the correlation generated by them $\beta:P(\mathcal{B}|\mathcal{X},\mathcal{Y})\to\mathbb{R}$. Generally, $\mathcal{Y}$ is considered to be a set of queries regarding $\mathcal{X}$ (see Fig.\ref{fig2}). Random Access Codes (RACs) are canonical examples of tasks that fit perfectly within this scenario. In a RAC task, the sender encodes a data set -- typically a string of input bits -- onto a physical system of bounded dimension and transmits it to the receiver, who then attempts to guess a randomly chosen part of the sender’s data set -- typically one of the sender’s input bits. While the protocol was first introduced by Wiesner (W) by the name conjugate coding \cite{Wiesner83}, the study drew renewed interest nearly two decades later when it was rediscovered by Ambainis, Nayak, Ta-Shma, and  Vazirani (ANTV) \cite{Ambainis99,Ambainis02}. 
\begin{figure}[h!]
\centering
\includegraphics[width=0.45\textwidth]{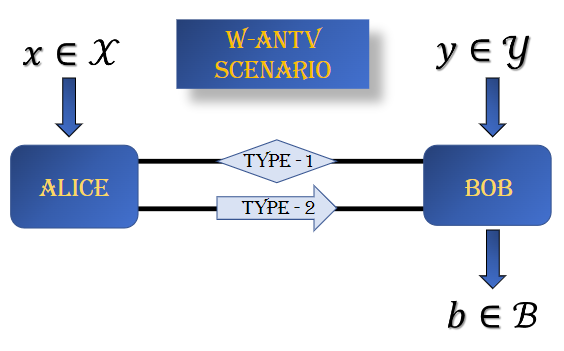}
\caption{(Color online) Wiesner-Ambainis- Nayak-Ta-Shma-Vazirani (W-ANTV) communication scenario. While Alice is given some classical input $x\in\mathcal{X}$, Bob is also given  (generally some query regarding $\mathcal{X}$) input $y\in\mathcal{Y}$. However, only Bob has to generate a classical output $b \in \mathcal{B}$. Like the H-FW case, two different types of resources might be used to optimize the payoff $\beta:P(\mathcal{B}|\mathcal{X},\mathcal{Y})\to\mathbb{R}$.}\label{fig2}
\end{figure}

Interestingly, in this scenario, the communication utility of a qubit (D$_c$-$2$ resource of {\bf Type-$2$}) has been shown to be more than a classical bit (D$_c$-$1$ resource of {\bf Type-$2$}), even in absence of any kind of {\bf Type-$1$} resources \cite{Wiesner83,Ambainis99,Ambainis02}. Subsequently, several variants of this task have been studied leading to interesting foundational implications \cite{Spekkens09,Banik15,Czekaj17,Ambainis19,Horodecki19,Saha19,Vaisakh21,Naik21}. Furthermore, in this scenario, it has also been shown that a resource of {\bf Type-$1$} can empower the communication utility of a resource of {\bf Type-$2$} \cite{Ambainis08,Pawlowski10,Tavakoli21,Bourennane}. In W-ANTV scenario, it has also been shown that post-quantum nonlocal correlations can empower classical communication arbitrarily, making communication complexity trivial \cite{vanDam,Brassard06,Buhrman10}. Nonetheless, all these works add more justification to the question that we aim to address in this present paper.

\subsection{Bell Scenario}
In this case both Alice and Bob are given classical random variables $x\in\mathcal{X}$ and $y\in\mathcal{Y}$, respectively; and both of them have to return some random variable $a\in\mathcal{A}$ and $b\in\mathcal{B}$, respectively to the Referee. Accordingly, they are given some payoff $\beta:P(\mathcal{A},\mathcal{B}|\mathcal{X},\mathcal{Y})\to\mathbb{R}$. Unlike the earlier two scenarios, they are not allowed to communicate with each other, {\it i.e.} resources of {\bf Type-$2$} are prohibited in this case. However, they can share some resources of {\bf Type-$1$}. This scenario is mostly studied within the quantum foundations community to understand the strength of different correlations. The seminal result of Bell \cite{Bell64,Bell66} (see also \cite{Mermin93}) establishes that this scenario is capable of revealing a hierarchy among different common past resources. In fact, from entangled quantum states, one can come up with correlations that yield more success for some suitably chosen payoffs than any classical shared randomness. Such correlations are more popularly known as Bell-nonlocal correlations (see \cite{Brunner14} for a review on Bell nonlocality). Subsequently, it has been shown that this scenario is capable of making a distinction between quantum and post-quantum correlations \cite{Popescu94,Cirelson80}. 
\begin{figure}[h!]
\centering
\includegraphics[width=0.45\textwidth]{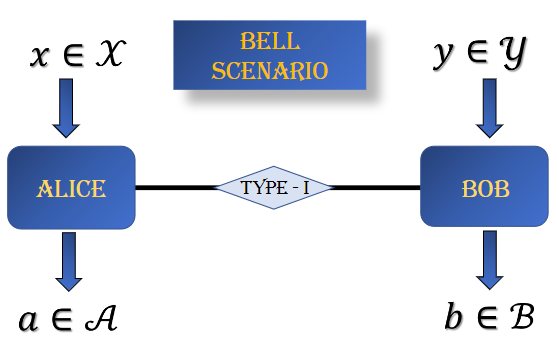}
\caption{(Color online) Canonical Bell Scenario. Here, both Alice and Bob are given inputs $x\in\mathcal{X}$ and  $y\in\mathcal{Y}$, respectively. Both of them need to return back classical outputs $a\mathcal{A}$ and $b\mathcal{B}$. Unlike the earlier two cases, no communication is allowed between Alice and Bob. However they can utilize {\bf Type-$2$} resources to optimize their joint payoff $\beta:P(\mathcal{A},\mathcal{B}|\mathcal{X},\mathcal{Y})\to\mathbb{R}$.}\label{three-Restaurant-game}
\end{figure}

More recently, this framework has been proved to be useful for separating different mathematical models for composite systems \cite{Slofstra20} which accordingly provides answers to some long-standing problems in complexity theory and operator theory \cite{Ji20} and leads to some undecidable consequences regarding the structure of  quantum logic \cite{Fritz21}. More recently, this scenario has been further generalized where Alice and Bob are given quantum inputs instead of classical random variables. \cite{Buscemi12,Branciard13,Lobo21}. 

Within this scenario, classical shared randomness that can produce only Bell-local correlations is considered to be free compared to the nonlocal correlations \cite{Wolfe20,Schmid20,Rosset20}. However, here our aim is to understand the utility of shared randomness in the H-FW scenario. In particular, whether it can enhance the utility of {\bf Type-$2$} resources in H-FW scenario. So it is judicious to start by considering shared randomness as a costly resource. In the following, we proceed to the main technical part of this work.   

\section{Classical superdense coding}
In this section, we will show how classical shared randomness can play a nontrivial role in enhancing the communication utility of a perfect classical communication channel. To this aim, we propose a class of two-party communication games and then analyze the payoffs of these games, first with $1$-bit of classical communication only and then with the additional assistance of shared randomness.  
\subsection{A two-party communication game}
Suppose Charlie is the manager of a chain of Restaurants. Let the total number of Restaurants be $n$. On each day one of the Restaurants remains closed and which one is to be closed is decided by Charlie randomly with uniform probability.

Charlie informs Alice which Restaurant is closed for the day. Another fellow, Bob, must visit one of Charlie's Restaurants for lunch each day. However, Bob does not know which Restaurant is closed for the day. He relies on Alice to communicate that information to him. Alice is restricted to communicating only 1-bit of information to Bob. The collective aim of Alice and Bob is to ensure that the following two conditions are satisfied:
\begin{itemize}
\item[(h$1$)] Bob never visits a closed Restaurant.
\item[(h$2$)] Bob visits each Restaurant with equal probability.   
\end{itemize}
The situation can be described with the following `visit' matrix:
\begin{align}\label{visit}
\mathbb{V}\equiv\begin{blockarray}{ccccc}
& 1_b & 2_b & \cdots & n_b \\
\begin{block}{c(cccc)}
1_c~~ & p(1_b|1_c) & p(2_b|1_c) & \cdots & p(n_b|1_c) \\
2_c~~ & p(1_b|2_c) & p(2_b|2_c) & \cdots & p(n_b|2_c) \\
\cdot~~ & \cdot & \cdot & \cdots & \cdot \\
\cdot~~ & \cdot & \cdot & \cdots & \cdot \\
n_c~~ & p(1_b|n_c) & p(2_b|n_c) & \cdots & p(n_b|n_c) \\
\end{block}
\end{blockarray}  
\end{align}
The entry $p(i_b|j_c)$ represents Bob's probability of visiting the $i^{th}$ Restaurant (hence the subscript `$b$') given that the $j^{th}$ Restaurant is closed (hence the subscript `$c$'). The condition ($\mathrm{h}1$) implies
\begin{align}
p(i_b|i_c)=0,~\forall~i\in\{1,\cdots,n\}, 
\end{align}
{\it i.e.} all the diagonal entries of the matrix $\mathbb{V}$ must be zero. Bob's probability of visiting the $i^{th}$ Restaurant can be obtained from the sum of the entries for the $i^{th}$ column of the matrix $\mathbb{V}$. The condition ($\mathrm{h}2$) thus reads as
\begin{align}
p(i_b)=\sum_{j=1}^np(i_b|j_c)p(j_c)=\frac{1}{n}\sum_{j=1}^np(i_b|j_c)=\frac{1}{n}.
\end{align}
Here, $p(j_c)=1/n,~\forall~j\in\{1,\cdots,n\}$, as the closing probabilities are assumed to be uniform and the last equality is imposed by the condition ($\mathrm{h}2$). Alice can send only $1$-bit of classical communication to Bob. Furthermore, we assume that shared randomness is a costly resource. Alice and Bob, however, are free to use local randomness on their part while playing the game. With this, we can formally define different classical strategies.
\begin{definition}[Classical deterministic strategy] \label{def:classical_deterministic_strategy}
 A classical deterministic strategy is an encoding-decoding tuple $(\mathrm{E},\mathrm{D})$, where $\mathrm{E}$ is a `$\log n$-bit to $1$-bit' deterministic function and  $\mathrm{D}$ is a `$1$-bit to $\log n$-bit' deterministic function, {\it i.e.} $\mathrm{E}:\{1,\cdots,n\}\mapsto\{0,1\}$ and $\mathrm{D}:\{0,1\}\mapsto\{1,\cdots,n\}$. 
\end{definition}
\begin{definition}
[Classical mixed strategy] A classical mixed strategy is a probabilistic strategy $\left(P_\mathbb{E},P_\mathbb{D}\right)$, where $P_\mathbb{E}$ and $P_\mathbb{D}$ are probability distributions over the space of deterministic encodings ($\mathbb{E}$) and decodings ($\mathbb{D}$), respectively. 
\end{definition}
\begin{definition}
[Classical correlated strategy] A classical correlated strategy is a probability distribution  $P_{\mathbb{E}\times\mathbb{D}}$ over the space of Cartesian product of deterministic encodings and decodings.
\end{definition}
Classical mixed strategies can be realized with local randomness on Alice's and Bob's sides. On the other hand, the implementation of classical correlated strategies requires Alice and Bob to posses classical shared randomness. Formally, classical shared randomness between two parties can be defined as a joint probability distribution $\{p(\alpha,\beta)\}$ on $\mathtt{A}\times\mathtt{B}$, where $\alpha\in\mathtt{A}$ and $\beta\in\mathtt{B}$ are random variables possessed by Alice and Bob, respectively. The amount of shared randomness can be quantified through classical mutual information, $I(\mathtt{A}:\mathtt{B}):=H(\mathtt{A})+H(\mathtt{B})-H(\mathtt{A},\mathtt{B})$, where $H(\mathtt{W}):=-\sum_wp(w)\log p(w)$ is the Shannon entropy of the random variable $w\in\mathtt{W}$. It is not hard to see that if Alice and Bob share the classically-correlated state $\rho=\frac{1}{2}\ket{00}\bra{00}+\frac{1}{2}\ket{11}\bra{11}\in\mathcal{D}(\mathbb{C}^2\otimes\mathbb{C}^2)$ then they can extract $1$-bit of classical shared randomness by performing $\sigma_z$ measurement on their respective subsystems. A similar process will yield $H(p)$-bit of shared randomness from the state $\rho=p\ket{00}\bra{00}+(1-p)\ket{11}\bra{11}$, with $p\in[0,1]$. For more elaborative discussions on classical shared randomness from a resource theoretic perspective we refer to the work \cite{Guha21}.  

At this point, we digress a bit to discuss mixed and correlated strategies in the light of Bayesian game theory. John Nash, in his seminal work, introduced the concept of Nash equilibrium and proved that any game with a finite number of actions for each player always has a mixed strategy Nash equilibrium \cite{Nash}. Later, Harsanyi introduced the notion of Bayesian games where each player has some private information unknown to other players \cite{Harsanyi}. Aumann proved that in the Bayesian game scenario the notion of mixed strategy Nash equilibrium needs to be generalized to a correlated equilibrium where some adviser provides advice to the players in the form of shared randomness to achieve the correlated equilibrium \cite{Aumann87}. Note that every pure/ mixed Nash equilibrium is also a correlated equilibrium, but the set of correlated equilibria is strictly larger than the set of mixed strategy Nash equilibria. It has also been shown that correlated equilibria are easier to compute \cite{Papadimitriou08}. Recently several interesting results have been reported connecting the study of Bayesian game theory and quantum nonlocality \cite{Brunner13,Pappa15,Roy16,Banik19(1)}.

In a communication task the following two questions are important to be explored. Firstly, whether there exists an optimal strategy to perform the task, and secondly, if it does exist, whether it is achievable with the resource available. For instance in Bayesian games where each player has some private information unknown to the other player although there exist the notion of Nash Equilibrium but its achievability is in question. Rather, a more general concept in this case is the correlated equilibrium which can be achieved if correlations are provided to the players as assistance. In communication scenarios, one player may have partial knowledge of the other player's information due to limited communication from the latter to the former. In this work, we focus on the existence of a winning strategy with classical communication, quantum communication, and the assistance of classical correlation (Shared-randomness). Considering shared randomness as a valuable resource, there may exist winning strategies that require its assistance. We consider the scenario where Charlie, as a referee, provides different resources to Alice and Bob to implement the winning strategy. Charlie allows arbitrarily large classical side channels from Alice to him and also from him to Bob. Through this Charlie will help Alice and Bob to develop the strategy that they will follow once the game starts. Apart from strategy development, the establishment of shared randomness between Alice and Bob will not be allowed through these side channels intervened by Charlie. In other words, we can say that the side channels are constrained. 
\begin{figure}[h!]
	\centering
	\subfloat[Step 1 : Announcement]{\includegraphics[width=0.45\linewidth]{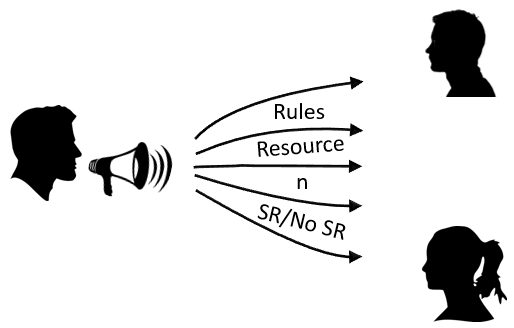}\label{fig17}}
 \hspace{1cm}
	\subfloat[Step 2 : Strategy Development]{\includegraphics[width=0.45\linewidth]{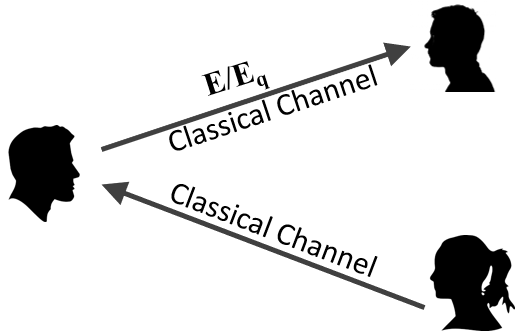}\label{fig18}}\\
 \subfloat[Step 3 : Active Gameplay]{\includegraphics[width=0.45\linewidth]{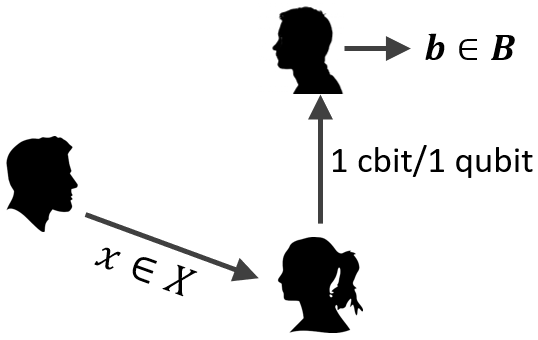}\label{fig19}}
 \caption{The restaurant game can be divided into three different steps. (a) In the first step, Charlie announces the rules of the game and also declares the resources that will provided to the players. (b) An arbitrarily large classical side channel is opened from Alice to Charlie and from Charlie to Bob with the constraint that Charlie will forward only the encoding strategies to Bob. Once the encoding strategies have been forwarded, these side channels will be permanently closed.  (c) In the active part of the game, Charlie will give $x\in\mathcal{X}$ to Alice and Alice sends $1$-cbit or $1$-qubit to Bob and accordingly Bob gives an output $b\in\mathcal{B}$.}
 \label{fig}
\end{figure}
For an explicit example, let's consider a specific two-party communication task: the Restaurant Game and their winning conditions (h1,h2) and explore the achievability of winning strategies. So the Game setup can be divided into three steps: Announcement, Strategy development, and Active Game-play (see Fig.\ref{fig}). 
\begin{itemize}
\item Step 1 : Announcement (see Fig.\ref{fig17})
\begin{itemize}
 \item [a.] Charlie is hosting a game where some payoff will be given depending on fulfilling certain winning conditions (h1 $\&$ h2).
\item[b.]  Two random players who have never met in the past, Sender (Alice) and Receiver (Bob), participate in the Restaurant game.
        \item[c.]Charlie declares number of restaurant (n) Bob has to visit in future and demands them to win $H_n(1/n)$ game.
        \item[d.] Charlie also declares the communication resources (classical channel/quantum channel) that will be opened in the future and announces whether some shared randomness will be provided as assistance.
    \end{itemize}
    \item Step 2 : Strategy Development (see Fig.\ref{fig18})
    \begin{itemize}
        \item [a.] 
        Charlie opens an arbitrarily large constrained classical channel from Alice to Charlie and from Charlie to Bob. But Charlie will communicate to Bob only the encoding strategy received from Alice.
        \item[b.] Alice chose an encoding strategy based on available resources and sent the encoding strategies to Bob via Referee.
        \item[c.] Charlie now closes both communication channels, from Alice to him and him to Bob.
       
    \end{itemize}
    \item Step 3 : Active Gameplay (see Fig.\ref{fig19})
    \begin{itemize}
        \item[a.] Charlie provides $x\in\mathcal{X}$ to Alice.  
        \item[b.] A one cbit/qubit communication channel has been opened from Alice to Bob.
         \item[c.] Bob outputs $b\in \mathcal{B}$ using some suitable decoding strategy.
    \end{itemize}
\end{itemize}
In this scenario shared randomness is a costly resource as the participating players are previously unknown to each other also they cannot generate shared randomness with the constrained classical channels. This setup explains the achievability of different winning strategies with the assistance of different communication resources but in the rest of the paper, we will focus on the existence of winning strategy.

Coming back to the winning strategy in our Restaurant game, Alice and Bob can follow any mixed strategy whenever $1$-bit of communication is allowed, but a correlated strategy requires additional resource of shared randomness. It is easy to see that the condition ($\mathrm{h}1$) can be readily satisfied within the limited set of mixed strategies. Alice just tells Bob if the first Restaurant is open or closed using $1$-bit classical message. If the first Restaurant is open, Bob visits the first Restaurant. Otherwise, he visits one of the other Restaurants.
However, satisfying the second condition is more tricky. In the next subsection, we analyze the cases where it can be satisfied. 

\subsection{Games winnable with mixed strategies}\label{Games winnable with mixed strategies}
We can consider a more general class of games where the winning condition ($\mathrm{h}2$) generalizes as follows:
\begin{itemize}
\item[($\mathrm{h}1^\prime$)] Bob never visits a closed Restaurant $\equiv(\mathrm{h}1)$.
\item[($\mathrm{h}2^\prime$)] Bob visits the $k^{th}$ Restaurant with probability $\gamma_k\geq0,~\forall~k \in \{1,\cdots,n\}$, where each $\gamma_k$ can in general be different.  
\end{itemize}
Values of $\gamma_k$'s are known to both Alice and Bob before the game starts. Since all Restaurants are closed with equal probability, we therefore have $\gamma_k\le\left(1-\frac{1}{n}\right),~\forall~k$. Clearly, each tuple $(\gamma_1,\gamma_2,\cdots,\gamma_n)$ defines a game which we will denote as $\mathbb{H}^n(\gamma_1,\gamma_2,\cdots,\gamma_n)$, and the special case where $\gamma_i$'s are uniform will be denoted as $\mathbb{H}^n\left(1/n\right)$. Our aim is to find which games are winnable under mixed classical strategies. A generic such strategy can be realized in the following steps.
\begin{itemize}
\item[(S-1)] If the $k^{th}$ Restaurant is closed, Alice tosses a $2$-sided biased coin having the outcomes $\{0,1\}$. The outcome probabilities of the coin are given by $P_k(0)=\alpha_k$ and $P_k(1)=1-\alpha_k$.
\item[(S-2)] Alice communicates the outcome of the coin toss to Bob through the perfect $1$-bit classical channel. 
\item[(S-3)] Bob prepares two $n$-sided coins with outcome probabilities specified by the probability vectors $\vec{r}=(r_1,r_2,\cdots,r_n)$ and $\vec{q}=(q_1,q_2,\cdots,q_n)$, respectively. If he receives $0$ from Alice he tosses the $\vec{r}$ coin and visits the $i^{th}$ Restaurant if $i^{th}$ outcome occurs. He follows a similar strategy with the $\vec{q}$ coin if $1$ is received from Alice.
\end{itemize}
With this strategy, the conditional probability $p(m_b|k_c)$ that Bob visits the $m^{th}$ Restaurant provided the $k^{th}$ Restaurant is closed turns out to be,
\begin{align}
p(m_b|k_c)=\alpha_k\times r_m+(1-\alpha_k)\times q_m.\label{eq1}
\end{align}
To satisfy the first condition that Bob never visits a closed Restaurant we must have,
\begin{align}
p(i_b|i_c)&=\alpha_i\times r_i+(1-\alpha_i)\times q_i=0\label{eq2}\\
&~~~~~~~~~~\forall~i~\in~ \{1,2,\cdots n\}.~\nonumber
\end{align}
Eq.(\ref{eq2}) holds for the $i^{th}$ Restaurant if and only if at least one of the following conditions is satisfied:
\begin{subequations}
\begin{align}
&\alpha_i=0~~\mbox{and}~~q_i=0,  \label{eq3}\\ 
&\alpha_i=1~~\mbox{and}~~r_i=0, \label{eq4}\\ 
&r_i=0~~\mbox{and}~~q_i=0.\label{eq5}
\end{align}
\end{subequations}
Without loss of generality we can divide the set of $n$ Restaurants in two set i.e.  $\mathbb{P}\equiv\{1,\cdots,a\}$ and
$\mathbb{O}\equiv\{a+1,\cdots,n\}$, such that $\gamma_i>0~\forall~i\in\mathbb{P}$ and 
$\gamma_i=0~\forall~i\in\mathbb{O}$. For every Restaurant either Eq.(\ref{eq3}) or Eq.(\ref{eq4}) or Eq.(\ref{eq5}) holds.
If $i\in \mathbb{P}$ then Eq.(\ref{eq5}) cannot hold for Restaurant $i$ since that would yield $\gamma_i=0$. Thus for $i\in \mathbb{P}$ either Eq.(\ref{eq3}) or Eq.(\ref{eq4}) must hold. Also we can conclude that for a Restaurant $i\in \mathbb{O}$  Eq.(\ref{eq5}) must hold. Therefore, the best they can do with a mixed classical strategy is to make a partition of Restaurants into three sets $X,~Y$ and $Z$ such that each of the Restaurants in sets $X,~Y$ and $Z$ satisfies Eq.(\ref{eq3}), Eq.(\ref{eq4}) and Eq.(\ref{eq5}), respectively. Alice sends $1$ if a Restaurant from the set $X$ is closed and sends $0$ if a Restaurant from the set $Y$ is closed. If a Restaurant $j\in Z$ is closed, she sends $0$ with  probability $\alpha_j$ and $1$ with probability $1-\alpha_j$. Bob never visits the Restaurants in $X$ whenever he receives $1$ and never visits the Restaurants in $Y$ whenever he receives $0$. Formally,
\begin{align*}
X\equiv\{j~|\alpha_j=0~~&\mbox{and}~~q_j=0\},\\ 
Y\equiv\{j~|\alpha_j=1~~&\mbox{and}~~r_j=0\},\\
Z\equiv\{j~|r_j=0~~&\mbox{and}~~q_j=0\}=\mathbb{O},\\
X\cup Y=\{1,2,\cdots, &a\}~~\&~~X\cap Y=\emptyset.
\end{align*}
Thus, the probability $p(m_b)=\sum_{k} p(m_b|k_c)p(k_c)$, that Bob visits $m^{th}$ Restaurant turns out to be,
\footnotesize
\begin{align}
p(m_b)&=\frac{1}{n} \left[\sum_{k\in X}p(m_b|k_c)+\sum_{k\in Y}p(m_b|k_c)+\sum_{k\in Z}p(m_b|k_c)\right]\nonumber\\
&=\frac{1}{n}\left[\sum_{k\in X}q_m+\sum_{k\in Y}r_m+\sum_{k\in Z}(\alpha_k\times r_m+(1-\alpha_k)\times q_m)\right].
\end{align}
\normalsize
According to the condition ($\mathrm{h}2^\prime$) we have $p(m_b)=\gamma_m$. and depending on whether $m$ belongs to $X$,$Y$ or $Z$ we have,
\begin{subequations}
\begin{align}
\gamma_{m}&=\frac{1}{n}\left[\sum_{k\in Y}r_m+\sum_{k\in Z}\alpha_k\times r_m\right]\nonumber\\
&=\left[\frac{|Y|}{n}+\bar{\alpha}_z\frac{|Z|}{n}\right]r_m,~~\forall~m~\in~X; \label{e1}\\
\gamma_{m}&=\frac{1}{n}\left[\sum_{k\in X}q_m+\sum_{k\in Z}(1-\alpha_k)\times q_m\right]\nonumber\\
&=\left[\frac{|X|}{n}+(1-\bar{\alpha}_z)\frac{|Z|}{n}\right]q_m,~~\forall~m~\in~Y; \label{e2}\\
\gamma_{m}&=0,~~\forall~m~\in~Z; \label{e3}
\end{align}
\end{subequations}
where $|\cdot|$ denotes the cardinality of a set, and $\bar{\alpha}_z:=\frac{1}{|Z|}\sum_{k\in Z}\alpha_k$. Thus, a classical winning strategy exists for the given game if and only if there exists a partition of $\mathbb{P}$ into two non-empty sets $X$ and $Y$, along with two probability vectors $\vec{r}$ and $\vec{q}$ such that
\begin{subequations}
\begin{align}
r_m&=\frac{n}{|Y|+\bar{\alpha}_z|Z|}\gamma_m,~~\forall~~m~\in~X,\\
q_m&=\frac{n}{|X|+(1-\bar{\alpha}_z)|Z|}\gamma_m,~~\forall~~m~\in~Y.
\end{align}
\end{subequations}
This defines the class of games that are winnable by  classical mixed strategies. From here on-wards we will specifically consider the games with $\gamma_1>0~\forall~i$. So it is better to see how Eq.(\ref{e1})-(\ref{e3}) get modified in such a case. Since we have $|Z|=0$, no Restaurant can lie in the set $Z$, thus we have
\begin{align}
\gamma_{m}=\frac{|Y|}{n}r_m,~~\forall~m\in X; ~~
\gamma_{m}=\frac{|X|}{n}q_m,~~\forall~m\in Y. \label{eq9}
\end{align}
Note that for the case where $\gamma_i>0~\forall~i$ sets $X$ and $Y$ must be non empty.
Now the probability vectors $\vec{r}$ and $\vec{q}$ are given by
\begin{align}
r_m&=\frac{n}{|Y|}\gamma_m,~\forall~m\in X;~~
q_m=\frac{n}{|X|}\gamma_m,~\forall~m\in Y.\label{f2}
\end{align}

\subsubsection{Special cases: $\mathbb{H}^n(1/n)$}
For this special case we have $|Z|=0$ and $a=n$. Using $\sum_{m\in X}r_m=1$ and $\sum_{m\in Y}q_m=1$ in Eq.(\ref{f2}) we have
\begin{align}
\sum_{m\in X}\gamma_m\frac{n}{|Y|}=1=\sum_{m\in Y}\gamma_m\frac{n}{|X|},\label{eq10}
\end{align}
which for $\mathbb{H}^n(1/n)$ game reduces to
\begin{align}
\label{eq11}
\sum_{m\in X}\frac{1}{|Y|}&=\sum_{m\in Y}\frac{1}{|X|},~~\Rightarrow~~
\frac{|X|}{|Y|}=\frac{|Y|}{|X|}.
\end{align}
This further implies $|X|=|Y|=n/2$. Thus, $\mathbb{H}^n(1/n)$ cannot be won in the case when $n$ is odd. However if  $n$ is even then the task becomes easy. We take  $X=\{1,2,...,n/2\}$ and  $Y=\{n/2+1,n/2+2,...,n\}$. Alice informs Bob if the closed Restaurant belongs to $X$ or $Y$. If the closed Restaurant is among $X$, then Bob tosses a n/2 -sided fair coin and visits one of the Restaurants in $Y$. If the closed Restaurant is in $Y$, then Bob tosses a n/2-sided fair coin and visits one of the Restaurants in $X$. Thus a $\mathbb{H}^n(1/n)$ game is always winnable using a mixed strategy if and only if $n$ is even. 

\subsubsection{Three-Restaurant general case:
$\mathbb{H}^3(\gamma_1,\gamma_2,\gamma_3)$}
\begin{center}
\begin{table}[b!]
\begin{tabular}{ |c||c|c|c|c|c|c|  }
\hline
~~Set~~ & $R_1$ & $R_2$ & $R_3$ & $R_4$ & $R_5$ & $R_6$\\
\hline\hline
X & $\{1\}$ & $\{2\}$ & $\{3\}$ & $~\{2,3\}~$ & $~\{1,3\}~$ & $~\{1,2\}~$\\
\hline
Y & $~\{2,3\}~$ & $~\{1,3\}~$ & $~\{1,2\}~$ & $\{1\}$ & $\{2\}$ & $\{3\}$\\
\hline
\end{tabular}
\caption{Six different partitioning of three Restaurants into two nonempty disjoint sets.}
\label{table2}
\end{table}
\end{center}
\vspace{-.7cm}
Note that in the three Restaurant case at most one of the $\gamma_i$ can be $0$. All games where one out of the three $\gamma_i$'s is $0$ are winnable with a classical mixed strategy. We provide explicit strategy for such cases. Consider that $\gamma_3=0$. Alice sends $0$ if Restaurant $1$ is closed, she sends $1$ if Restaurant $2$ is closed and if Restaurant 3 is closed she sends $0$ with probability $p$ and $1$ with probability $1-p$. Bob visits Restaurant $2$ if he receives a $0$ and he visits Restaurant $1$ if he receives a $1$. This strategy ensures that Bob never visits a closed Restaurant and we also have
\begin{align}
\gamma_1=\frac{1}{3}[1+(1-p)],~~\gamma_2=\frac{1}{3}[1+p],~~\gamma_3=0.
\end{align}
With an appropriate choice of the values for $p\in[0,1]$ any game with $\gamma_3=0$ can be won. Similar strategies hold for the cases where $\gamma_1=0$ and $\gamma_2=0$. We now move on to the scenario where $\gamma_i>0~\forall~i$. We start by looking at all the different partitions $X$ and $Y$ of the Restaurants. Since the set $X$ and $Y$ must be non-empty we can make six different partitions of three Restaurants as shown in Table \ref{table2}. Without loss of any generality, we can consider only the three cases $R_1$, $R_2$, and $R_3$. For these cases we have,
\begin{itemize}
\item[$R_1:$] In this case we have, $r_2=r_3=q_1=0, r_1=1$ and $q_2+q_3=1$; that imply, $\gamma_1=\frac{2}{3}$ and $\gamma_2+\gamma_3=\frac{1}{3}$.
\item[$R_2:$] Here, $r_1=r_3=q_2=0, r_2=1$ and $q_1+q_3=1$; which further imply, $\gamma_2=\frac{2}{3}$ and $\gamma_1+\gamma_3=\frac{1}{3}$.
\item[$R_3:$] In this case, $r_1=r_2=q_3=0, r_3=1$ and $q_1+q_2=1$; and consequently, $\gamma_3=\frac{2}{3}$ and $\gamma_1+\gamma_2=\frac{1}{3}$.
\end{itemize}
The other three can be achieved from these three by  inverting Alice's encoding $0\rightarrow 1$ and $1\rightarrow 0$. Therefore, only those $\mathbb{H}^3(\gamma_1,\gamma_2,\gamma_3)$ are perfectly winnable with a classical mixed strategy where one of the $\gamma_i$ value is $2/3$ or $0$ (see Fig. \ref{fig4}). In particular, the game $\mathbb{H}^3(1/3)$ is not winnable by a classical strategy. 
\begin{figure}[t!]
\centering
\includegraphics[width=0.5\textwidth]{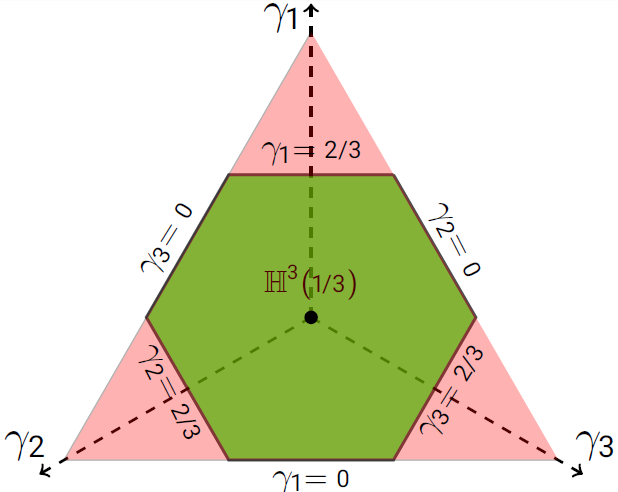}
\caption{(color online) Parameter-space ($\gamma_1+\gamma_2+\gamma_3=1$ plane) of the games $\mathbb{H}^3(\gamma_1,\gamma_2,\gamma_3)$. Orange shaded regions are the unphysical games as the conditions ($\mathrm{h}1^\prime)$ and ($\mathrm{h}2^\prime)$ cannot be satisfied for the parameters values chosen from there. The green shaded region (polytope) are all games winnable with correlated strategies with $1$-bit of shared randomness. The boundaries of the green polytope, {\it i.e.} $\gamma_i=0~\mbox{and}~2/3,~\mbox{for}~i\in\{1,2,3\}$, are the only games winnable with classical mixed strategies.}\label{fig4}
\end{figure}

\subsubsection{Even Restaurant: unwinnable game}
We have already seen that all the games $\mathbb{H}^n(1/n)$ are perfectly winnable with classical mixed strategies whenever $n$ is even. Naturally, the question arises whether all even Restaurant games are winnable with such strategies when generic cases are considered. We answer this question negatively by constructing an explicit example of such a game. Note that Eq.(\ref{eq9}) yields,
\begin{align}
\sum_{m\in X}\gamma_m=\frac{|Y|}{n},~~~\&~~~\sum_{m\in Y}\gamma_m=\frac{|X|}{n}.\label{eq13}   
\end{align}
Now it becomes easy to come up with an even Restaurant game that is impossible to win using a mixed strategy. For instance, consider the $4$-Restaurant game specified by $\gamma_1=2/5$ and $\gamma_2=\gamma_3=\gamma_4=1/5$. No non-empty disjoint partitioning of the four Restaurants satisfies the conditions in Eq.(\ref{eq13}). Thus the game $\mathbb{H}^4\left(\frac{2}{5},\frac{1}{5},\frac{1}{5},\frac{1}{5}\right)$ is not winnable with classical mixed strategies.

\subsection{Games winnable with classical correlated strategies}\label{1bit1SR}

We have already seen that only a subclass of games $\mathbb{H}^n(\gamma_1,\cdots,\gamma_n)$ can be won with mixed strategies. If Alice and Bob are allowed to share classical shared randomness, then they can follow any correlated strategy and can win any of the game $\mathbb{H}^n(\gamma_1,\cdots,\gamma_n)$. For instance, consider the $3$-Restaurant case first. Alice and Bob can share a random variable $\lambda \in \{1,2,3\}$ which determines the partitions $R_\lambda$ that they use in a particular run. Using $R_1$ they can win all games of the form $\left(\gamma_1=\frac{2}{3},\gamma_2,\gamma_3\right)$, using $R_2$ they can win all games of the form $\left(\gamma_1,\gamma_2=\frac{2}{3},\gamma_3\right)$, and using $R_3$ they can win all games of the form $\left(\gamma_1,\gamma_2,\gamma_3=\frac{2}{3}\right)$. Therefore, using shared randomness they can win all games of the form,
\begin{align*}
r_1 \left(\frac{2}{3},a_2,a_3\right)+r_2\left(b_1,\frac{2}{3},b_3\right)+r_3\left(c_1,c_2,\frac{2}{3}\right),
\end{align*}
where $a_2+a_3=b_1+b_3=c_1+c_2=\frac{1}{3}$ and $r_1+r_2+r_3=1$ and these cover all possible games $\mathbb{H}^3(\gamma_1,\gamma_2,\gamma_3)$. 

The above strategy uses $\log3$ bits of shared randomness. However, a more efficient protocol is possible if we look into the geometry of the game space. The set of all possible games $\mathbb{H}^3(\gamma_1,\gamma_2,\gamma_3)$ forms a polytope embedded in the 2-simplex as shown in Fig.\ref{three-Restaurant-game}. All games that are winnable through a strategy corresponding to the partition $R_1$ form the facet $\gamma_1=\frac{2}{3}$ of the green polytope. Similarly, partitions $R_2$ and $R_3$ form facets that correspond to $\gamma_2=\frac{2}{3}$ and $\gamma_3=\frac{2}{3}$, respectively. Any game lying within the green polytope can be won by taking a suitable convex mixture of two points from the facets $\gamma_i=0$ and $\gamma_i=\frac{2}{3}$; $i\in\{1,2,3\}$. Therefore $1$-bit of shared randomness along with $1$-bit classical channel suffice to construct a perfect strategy for any of the games $\mathbb{H}^3(\gamma_1,\gamma_2,\gamma_3)$. 

Moving to the general case, any of the games $\mathbb{H}^n(\gamma_1,\cdots,\gamma_n)$ can be perfectly won with $1$-bit of classical communication if sufficient amount of shared randomness is allowed. A trivial upper bound can be obtained from the convex structure of game space. For general $n$, the game space is a convex polytope embedded in $\mathbb{R}^n$. The extreme points of this polytope are the games specified by the vectors of the form $\left(\frac{n-1}{n},\frac{1}{n},0,\cdots,0\right)\in\mathbb{R}^{n}$, and its all possible permutations, total $n(n-1)$ in number. All such games can be won with just $1$-bit of classical communication without requiring any shared randomness. For instance, while playing the game $\mathbb{H}^n\left(\frac{n-1}{n},\frac{1}{n},0,\cdots,0\right)$, Alice will send `$0$' if the first Restaurant is closed, else she sends `$1$'. Bob visits the second Restaurant when he receives `$0$', else he visits the first Restaurant. It is not hard to see that both the conditions ($\mathrm{h}^\prime1$) and ($\mathrm{h}^\prime2$) are satisfied with this strategy. Since any vector $(\gamma_1,\cdots,\gamma_n)\in\mathbb{R}^{n}$ specifying the game $\mathbb{H}^n(\gamma_1,\cdots,\gamma_n)$ can always be expressed as a convex mixture of aforesaid $n(n-1)$ extreme points, therefore $\log (n^2-n)$-bit of shared randomness, along with $1$-bit of communication, will suffice for winning any of the games $\mathbb{H}^n(\gamma_1,\cdots,\gamma_n)$. However, this seems an extreme overestimate, and we believe that the games can be won with much less amount of shared randomness. In fact, it would be quite interesting to show that any such game can be won with $1$-bit of communication when aided with just $1$-bit of shared randomness.  The above result can be summarized as the following theorem.
\begin{theorem}\label{theo1}
Shared randomness can increase the utility of a perfect 1 cbit classical channel in H-FW scenario.
\end{theorem}

A game $H^n(\gamma_1,\cdots,\gamma_n)$ winnable with some classical correlated strategy but not winnable with any classical mixed strategy establishes the fact that classical shared randomness can empower utility of a perfect but limited classical communication line in the Holevo \& Frenkel-Weiner kind of scenario. Important to note that the single-shot utility of a perfect classical channel gets empowered with shared randomness. This, in a sense, can be thought as a classical version of `quantum superdense coding' phenomenon where also single-shot communication utility of a perfect quantum channel gets enhanced with preshared entanglement. Of course, there is a crucial difference between these two. In the case of quantum superdense coding channel's utility is quantified through classical mutual information between senders and receivers data which can be increased with preshared entanglement. In the classical scenario, shared randomness cannot increase the mutual information. Therefore in the classical scenario, the usefulness of shared randomness in enhancing the single-shot utility of a perfect channel must be analyzed with some payoff different from classical mutual information. The winning condition of our Restaurant games stands for one such payoff. Our Restaurant game is only a particular example of tasks that establish such nontrivial communication utility of classical shared randomness. Nonetheless, it motivates further research to explore such novel usefulness of shared randomness in single-shot paradigm. In the following we, however, proceed to analyze the Restaurant games with quantum resource. 

\section{Communication advantage of quantum system} \label{quantum strategy}
When playing the Restaurant games in quantum scenario, Alice can send a qubit system to Bob instead of a classical bit. Like classical mixed strategies, no classical shared randomness is allowed between Alice and Bob. In such a case, a quantum strategy can be defined as follows.     
\begin{definition}
[Quantum strategy] A quantum strategy is a encoding-decoding tuple $(\mathrm{E}_q,\mathrm{D}_q)$, where $\mathrm{E}_q$ is a `$\log n$-bit to $1$-qubit' bijective function and  $\mathrm{D}_q$ is a $n$ outcome POVM ,{\it i.e.} $\mathrm{E}_q:i\mapsto\rho_i\in\mathcal{D}(\mathbb{C}^2)$, with $i\in\{1,\cdots,n\}$; and $\mathrm{D}_q:\{\pi_j~|~\pi_j\ge0~\&~\sum_{j=1}^n\pi_i=\mathbb{I}\}$, where $\mathbb{I}$ is the identity operator on $\mathbb{C}^2$.
\end{definition}
Bob makes the decision to visit the Restaurants based on his measurement outcomes. Here we assume that the quantum communication line is perfect, {\it i.e.} the qubit state, Alice intends to send Bob, does not get interrupted with noise. Noisy case analysis we defer to the later section. Interestingly, we will show that there exist quantum strategies that are advantageous over classical mixed strategies. 

\subsection{Special case: $\mathbb{H}^n(1/n)$} \label{3_Restaurant__quantum_strategy_special_case}
Let us first consider the case $n=3$. For encoding, Alice chooses a symmetric set of pure states lying on an equilateral triangle from a great circle of the Bloch sphere. When the $k^{th}$ Restaurant is closed, she sends the state $\rho_k=\ket{\psi_k}\bra{\psi_k}=\frac{1}{2}\left[\mathbb{I}+\hat{n}_k\cdot \vec{\sigma}\right]$, where  $\hat{n}_k$ is the Bloch vector $\left(\sin\frac{2\pi (k-1)}{3},0,\cos\frac{2\pi (k-1)}{3}\right)^{\mathrm{T}}\in\mathbb{R}^3$; $k\in\{1,2,3\}$ and $\mathrm{T}$ denotes transposition. For decoding, Bob performs the measurement $\mathcal{M}(3)\equiv\left\{\pi_k:=\frac{1}{3}\left[\mathbb{I}-\hat{n}_k\cdot \vec{\sigma}\right]\right\}_{k=1}^3$, and visits the $k^{th}$ Restaurant if outcome corresponding to the effect $\pi_k$ is observed. With this strategy we have $p(i_b|i_c)=\tr[\rho_i\pi_i]=0,~\forall~i$, ensuring the condition ($\mathrm{h}1$). Furthermore, Bob's probability of visiting $m^{th}$ Restaurant turns out to be,
\footnotesize
\begin{align}
p(m_b) &= \frac{1}{3}\sum_{k} p(m_b|k_c)= \frac{1}{3}\sum_{k} \Tr\left[\frac{1}{2}\left[\mathbb{I}+\hat{n}_k\cdot \vec{\sigma}\right]\frac{1}{3}\left[\mathbb{I}-\hat{n}_m\cdot \vec{\sigma}\right]\right] \nonumber\\
&= \frac{1}{9}\sum_{k} \left[1-\hat{n}_k\cdot\hat{n}_m\right] = \frac{1}{9}\left[3-\left(\sum_k\hat{n}_k\right)\cdot\hat{n}_m\right]=\frac{1}{3},\nonumber
\end{align}
\normalsize
\begin{figure}[t!]
\centering
\includegraphics[width=0.45\textwidth]{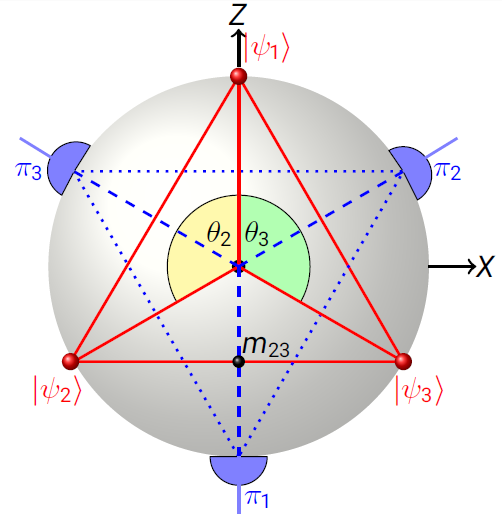}
\caption{(Color online) Quantum strategy for $\mathbb{H}^3(\gamma_1,\gamma_2,\gamma_3)$. Alice sends the qubit $\ket{\psi_k}$ when the $k^{th}$ Restaurant is closed. Red dots denote Alice's encodings. Blue half-circles denote the rank-one effects corresponding to  Bob decoding. Blake dot denotes the mid-point $m_{23}$ of the cord joining the Bloch vectors $\psi_2~\&~\psi_3$. Here the strategy is shown for $\gamma_1=\gamma_2=\gamma_3=1/3$, {\it i.e.}, for the game $\mathbb{H}^3(1/3)$. For other cases, $\theta_2$ and $\theta_3$ need to be varied accordingly.}\label{fig5}
\end{figure}
which ensures the condition ($\mathrm{h}2$). The strategy can be generalized for any of the games $\mathbb{H}^n(1/n)$ with odd $n$. Recall that, for even $n$ the game $\mathbb{H}^n(1/n)$ is winnable with classical mixed strategy. For an arbitrary odd $n$, Alice uses the encoding $\rho_k=\frac{1}{2}\left[\mathbb{I}+\hat{n}_k\cdot \vec{\sigma}\right]$, where $\hat{n}_k$'s are symmetrically distributed on the equatorial plane. Bob performs the measurement $\mathcal{M}(n)\equiv\left\{\pi_k:=\frac{1}{n}\left[\mathbb{I}-\hat{n}_k\cdot \vec{\sigma}\right]\right\}_{k=1}^n$ and visits the $k^{th}$ Restaurant if outcome corresponding to the effect $\pi_k$ is observed. A similar calculation as before ensures both the conditions ($\mathrm{h}1$) and ($\mathrm{h}2$).

\subsection{Three-Restaurant general case:
$\mathbb{H}^3(\gamma_1,\gamma_2,\gamma_3)$}

In this general case, to fulfill the condition ($\mathrm{h}1^\prime$), Alice must choose some pure states for encoding. Let she sends the state $\ket{\psi_i}$ to Bob when the $i^{th}$ Restaurant is closed. To satisfy the condition ($\mathrm{h}1^\prime$) Bob must perform a decoding measurement of the form $\mathcal{M}=\left\{\alpha_i\ketbra{\psi_i^{\perp}}{\psi_i^{\perp}}~|~\alpha_i>0~\&~\sum_i\alpha_i=2\right\}_{i=1}^3$ and he visits the $i^{th}$ Restaurant if the $i^{th}$ effect clicks. To satisfy this requirement the completely mixed state must lie within the triangle formed by Bloch vectors corresponding to the encoding states $\{\ket{\psi_i}\}_{k=1}^3$. Without loss of any generality, Alice can choose her encodings as $\psi_1=(0,0,1)^{\mathrm{T}},~\ket{\psi_2}=(-\sin\theta_2,0,\cos\theta_2)^{\mathrm{T}}$ and $\ket{\psi_3}=(\sin\theta_3,0,\cos\theta_3)^{\mathrm{T}}$; where $\psi_i$ is the Bloch vector of the state $\ket{\psi_i}$ and $\theta_2,\theta_3\in[0,\pi]$\footnote{Note that $\theta_2+\theta_3$ cannot be lesser than $\pi$,  since this configuration would not form a valid measurement.} are the polar angles for the corresponding Bloch vectors (see Fig.\ref{fig5}). Accordingly, the requirements of the condition ($\mathrm{h}1^\prime$) read as,
\begin{figure}[b!]
\centering
\includegraphics[width=0.4\textwidth]{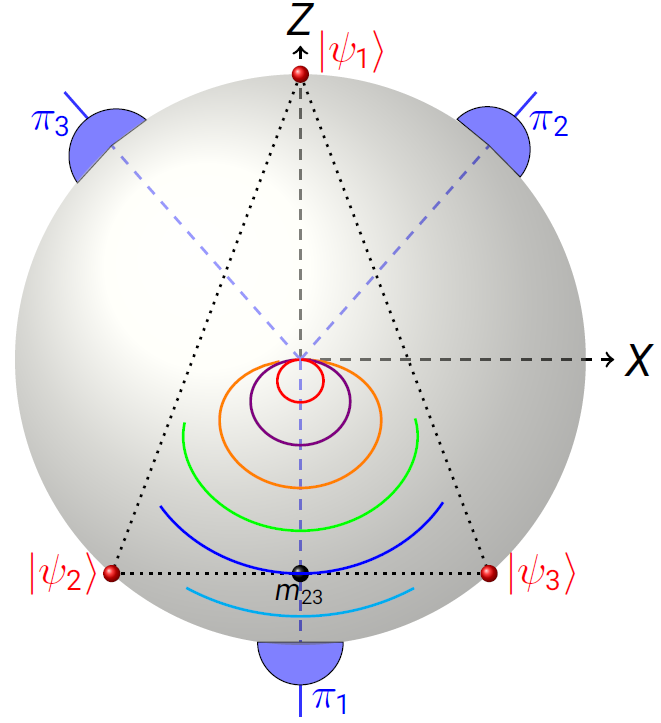}
\caption{(Color online) This figure shows locus of the midpoints $m_{23}$ having constant $\gamma_1$. Once $\gamma_1$ is fixed, $m_{23}$ completely specifies the value of $\gamma_2$ and $\gamma_3$. Thus it is sufficient to plot constant $\gamma_1$ curves. Here, $\gamma_1=0.6,~0.5,~0.4,~0.3,~0.2, \&~0.1$ curves are plotted. The black dotted triangle shows an explicit strategy for the game $\mathbb{H}^3\left(0.5,0.25,0.25\right)$. As the black dot moves on the blue curve we get strategies for the games of the form $\mathbb{H}^3\left(0.5,\frac{p}{2},\frac{1-p}{2}\right)$, with $p\in[0,1]$. The leftmost point on the blue curve corresponds to the game $\mathbb{H}^3\left(0.5,0,0.5\right)$ while the rightmost point corresponds to $\mathbb{H}^3\left(0.5,0.5,0\right)$.}\label{fig6}
\end{figure}
\begin{subequations}
\begin{align}
\alpha_1+\alpha_2+\alpha_3&=2,\\
\alpha_1+\alpha_2\cos\theta_2+\alpha_3\cos\theta_3&=0,\\
-\alpha_2\sin\theta_2+\alpha_3\sin\theta_3&=0.
\end{align}
\end{subequations}
These further lead to,
\begin{subequations}
\begin{align}
\label{eq16}
\alpha_1&=\frac{2\sin(\theta_2+\theta_3)}{\sin(\theta_2+\theta_3)-\sin\theta_2-\sin\theta_3},\\
\alpha_2&=\frac{-2\sin\theta_3}{\sin(\theta_2+\theta_3)-\sin\theta_2-\sin\theta_3},\\
\alpha_3&=\frac{-2\sin\theta_2}{\sin(\theta_2+\theta_3)-\sin\theta_2-\sin\theta_3},
\end{align}
\end{subequations}
and accordingly, we have,
\begin{align}
\label{eq17}
\gamma_1&= \frac{1}{3}(p(1|2)+p(1|3))\nonumber\\
&= \frac{1}{3}\Tr[\alpha_1\ketbra{\psi_1^{\perp}}{\psi_1^{\perp}}\left(\ketbra{\psi_2}{\psi_2}+\ketbra{\psi_3}{\psi_3}\right)]\nonumber\\
&= \frac{1}{3}\frac{\sin(\theta_2+\theta_3)(2-\cos\theta_2-\cos\theta_3)}{(\sin(\theta_2+\theta_3)-\sin\theta_2-\sin\theta_3)}.
\end{align}
The aforesaid encoding can be uniquely specified by specifying the state $\ket{\psi_1}=\ket{0}$ and fixing the midpoint $m_{23}$ of the line joining Bloch vectors of $\ket{\psi_2}$ and $\ket{\psi_3}$ (see Fig.\ref{fig5}). This is because for any point within a great circle (except the center) there exists a unique chord having that point as the midpoint of the chord. Thus choosing the encoding state $\ket{\psi_1}=\ket{0}$ when the first Restaurant is closed, the complete encoding is specified just by the midpoint $m_{23}$, where the Bloch vector corresponding to the Restaurant $2$ lies on the left side. Now, Eq.(\ref{eq17}) can be used to plot locus of the midpoints $m_{23}$ with constant $\gamma_1$. Such a plot is shown in Fig.\ref{fig6}. We thus conclude that all games $\mathbb{H}^3(\gamma_1,\gamma_2,\gamma_3)$ are perfectly winnable through some quantum strategy. This leads us to the following theorem.
\begin{theorem}\label{theo2}
    In the absence of any preshared correlation utility of 1 qubit communication is higher than 1 bit of classical communication in H-FW scenario.
\end{theorem}

\subsection{Quantum strategy for  $\mathbb{H}^4\left(\frac{2}{5},\frac{1}{5},\frac{1}{5},\frac{1}{5}\right)$}
In case of any four Restaurant game of the form $\mathbb{H}^4\left(\gamma_1,\frac{1-\gamma_1}{3},\frac{1-\gamma_1}{3},\frac{1-\gamma_1}{3}\right)$ with $\gamma_1>0$, Alice sends the state $\ket{\psi_i}$ to Bob when the $i^{th}$ Restaurant is closed. Without loss of generality, she can choose $\ket{\psi_1}\equiv(0,0,1)^{\mathrm{T}}$, and due to symmetry of the visiting probability ($\gamma_2=\gamma_3=\gamma_4$) the encoding of other three state will orient symmetrically on a constant $z$ plane. One can orient $\ket{\psi_2}$ along the $x$ axis on that constant $z$ plane. However, once $\ket{\psi_2}$ is chosen the encodings for $\ket{\psi_3}$ and $\ket{\psi_4}$ are fixed. Let Alice's encodings are 
\begin{align*}
1&\mapsto \ket{\psi_1}\equiv(0,0,1)^{\mathrm{T}},\\
2&\mapsto \ket{\psi_2}\equiv\left(\sin{\theta},0,\cos{\theta}\right)^{\mathrm{T}},\\
3&\mapsto \ket{\psi_3}\equiv\left(-\frac{1}{2}\sin\theta,\frac{\sqrt{3}}{2}\sin\theta,\cos\theta\right)^{\mathrm{T}},\\
4&\mapsto \ket{\psi_4}\equiv\left(-\frac{1}{2}\sin\theta,-\frac{\sqrt{3}}{2}\sin\theta,\cos\theta\right)^{\mathrm{T}}.
\end{align*}
To satisfy the condition ($\mathrm{h}1^\prime$), Bob must perform a decoding measurement of the form $\mathcal{M}=\left\{\alpha_i\ketbra{\psi_i^{\perp}}{\psi_i^{\perp}}~|~\alpha_i>0~\&~\sum_i\alpha_i=2\right\}_{i=1}^4$ and he visits the $i^{th}$ Restaurant when  $i^{th}$ effect clicks. With this encodings and decoding, the condition ($\mathrm{h}2^\prime$) further demands
\begin{subequations}
\begin{align}
\alpha_1+\alpha_2+\alpha_3+\alpha_4=&2,\\
(2\alpha_2-\alpha_3-\alpha_4)\times\sin\theta=&0,\\
(\alpha_3-\alpha_4)\times\sin\theta=&0,\\
\alpha_1+(\alpha_2+\alpha_3+\alpha_4)\times\cos\theta=&0,\\
(1-\cos\theta)\times\alpha_1=&\frac{8}{3}\gamma_1.
\end{align}
\end{subequations}
Solving this set of equations for $\gamma_1=\frac{2}{5}$ we obtain $\alpha_1=\frac{16}{23},\alpha_2=\alpha_3=\alpha_4=\frac{10}{23}$ and $\cos\theta=-\frac{8}{15}$, which completely specify the perfect strategy with qubit communication.

This shows only a particular example where qubit communication yields the perfect strategy for the four-Restaurant case. It might be interesting to analyze the general four-Restaurant scenario $\mathbb{H}^4(\gamma_1,\gamma_2,\gamma_3,\gamma_4)$. In fact, it would be worth interesting to classify the games $\mathbb{H}^4(\gamma_1,\cdots,\gamma_n)$ that allow perfect qubit strategies. This question we left for future research. In the next, we rather move to analyze the source of the obtained advantage from a more foundational perspective. 

\section{Origin of the advantage} 
Quantum advantages are hard to find and even harder to prove. In this section, we will analyze which particular non-classical features underlies to the aforesaid communication advantage we have obtained. We will also analyze whether similar advantages are possible or not with some hypothetical non-classical toy systems.  
\subsection{Two no-go results in quantum scenario} \label{sec:quantum_no_go}
In search of the non-classical features of quantum theory that exhibit the above advantage in the communication scenario, here we present two important \textit{no-go} results.
\begin{proposition}
A qudit is no better than a c-dit, in H-FW scenario, if Alice encodes the inputs in d-dimensonal commuting density operators.
\end{proposition}
\begin{proof}
Let us consider that Alice uses \textit{n} commuting density operator, {\it i.e.} $\{\rho_1,\rho_2,\cdots,\rho_n\}$, where $\rho_i\in \mathcal{D}(\mathbb{C}^d)$ for her encoding. let the common eigenstate of the commuting density operators be $\{\ket{\phi_1},\ket{\phi_2},\cdots,\ket{\phi_d}\}$, then the encoded density operators can be express as $\rho_i=\sum_{j=1}^d C_{ij}\ket{\phi_j}\bra{\phi_j}$. After receiving the encoded qudit, the receiver performs an \textit{n}-outcome positive operator valued measure (POVM) $\mathcal{M}\equiv\{E_{k}: \sum_{k}E_{k}=\mathbb{I}\}_{k-1}^n$ to decide among \textit{n}-possibilities. The probability of making $k^{\text{th}}$ decision by the receiver is $p(E_{k}|i)=\text{Tr}(E_{k}\rho_i)=\sum_{j=1}^dC_{ij}\text{Tr}(E_{k}\ketbra{\phi_j}{\phi_j})=\sum_{j=1}^dC_{ij}~p(E_{k}|\phi_j)$, when he receives the state $\rho_i$; $i\in\{0,1,\cdots,n\}$.  In the analogous classical scenario, depending upon the input received from Refree, Alice chooses one of the \textit{d}-faced coin $\mathcal{C}_{i}$ with the bias $C_{ij}$ where $i\in\{1,2,\cdots,n\};j\in\{1,2,\cdots,d\}$. Depending on the outcome of the coin flip $\mathcal{C}_{i}$, Alice will send Bob a classical digit ranging from 0 to d, instead of the qudit $\{\rho_1,\rho_2,\cdots,\rho_d\}$. Bob will then choose one of the n-faced coins $\mathcal{\Tilde{C}}_{i}$ based on the digit he received from Alice. Each of these coins has a bias $\{p(E_{1}|\phi_i),\cdots,p(E_{n}|\phi_i)\}$, where $i\in\{1,2,\cdots,n\}$. Therefore, the probability generated by commuting qudits encoding can always be simulated by a d-bit classical communication. This completes the proof.
\end{proof}

The above result proves the necessity of non-commuting encoding to exhibit the advantage of qudit communication over the $d$-bit classical channel in the Restaurant game $\mathbb{H}^n(\gamma_1,\cdots,\gamma_n)$. Our next result deals with the decoding part at Bob's end.
\begin{proposition}
An advantage of a qudit over its classical counterpart in the H-FW scenario requires the non-commutativity of the measurement operators at the receiver's end for decoding. 
\end{proposition}
\begin{proof}
Let the sender uses one among $n$-qudits $\{\ket{\psi_{1}},\ket{\psi_{2}},\cdots,\ket{\psi_{n}}\}$ to encode her messages and sends the encoded state to the receiver. Consider that Bob performs a \textit{n}-outcome positive operator valued measure (POVM) $\mathcal{M}\equiv\{E_{k}: \sum_{k}E_{k}=\mathbb{I}\}_{k-1}^n$ to decide among \textit{n}-possibilities, where all the measurement operators {\it i.e.} $E_k$ commute. So the measurement operators can be expressed as $E_k=\sum_{i=1}^dC_{ki}\ket{\phi_i}\bra{\phi_i}$, where $\ket{\phi_i}$ are the common eigenstate of the measurement operators $E_k$. The encoding states $\ket{\psi_j}$ can be expressed in this eigen basis {\it i.e.} $\ket{\psi_j}=\sum_{i=1}^d\alpha_{ji}\ket{\phi_i}$.  In the classical counterpart, the sender can replace the quantum state $\ket{\psi_{j}}$ by a source $S_{j}$ generating classical random variables $\{0,1,\cdots,n\}$ with probabilities $\{|\alpha_{j1}|^{2},|\alpha_{j2}|^{2},\cdots,|\alpha_{jd}|^{2}\}$. After receiving the classical random variable Bob will choose one of the \textit{n}-faced coins $\mathcal{C}_i$ with $C_{ki}$ bias and the receiver can simulate the quantum probabilities obtained through commuting decoding. This proves the claim.
\end{proof}

The quantum advantage in our Restaurant game, therefore, must invoke `quantum interference' in the form of quantum superposition at encoding and noncommutivity at decoding steps. This is quite similar the communication advantage as obtained in \cite{Ambainis99,Ambainis02,Spekkens09,Banik15,Czekaj17,Ambainis19,Horodecki19,Saha19,Vaisakh21,Naik21} within the W-ANTV scenario. Here, non-commutative measurements, more specifically measurements that are non-jointly measurable, play a crucial role in the decoding step at the receiver's end. Within the H-FW scenario, the entanglement assisted advantage of a perfect classical channel as reported in \cite{Frenkel21} once again uses non-compatible measurements at the decoding step (see Appendix \ref{appendix-a}). On the other hand, the seminal 'quantum superdense coding' protocol \cite{Bennett92} makes smart uses of quantum entanglement both at encoding and decoding (Bell basis measurement) steps (see also \cite{Mermin02}).          

 \subsection{Is the advantage strictly quantum?}\label{polygon}
This question is quite important from a foundational perspective. Recall that while the Seminal result of J. S. Bell establishes `nonlocal' behavior of quantum correlations \cite{Bell66}, later Popescu and Rohrlich report such correlations which are beyond quantum in nature \cite{Popescu94}. On the other hand, communication advantages of a qubit over its classical counterpart in the W-ANTV kind of scenario can also be obtained with a hypothetical non-classical, in fact with better than quantum success \cite{Banik15}. In a similar way, the advantage reported in the present paper can also be obtained with hypothetical non-classical systems known as polygon theories. 

\subsubsection{Polygon theories} This class of models can be specified by the tuple $\mathcal{P}_{ly}(n)\equiv\left(\Omega(n),\mathcal{E}(n)\right)$ \cite{Janotta11}, where $\Omega(n)$ is the state space and $\mathcal{E}(n)$ is the effect space. For a fixed $n$, $\Omega(n)$ is the convex hull of $n$ pure states $\{\omega_i\}_{i=1}^{n}$, where $\omega_i:= \left(r_n\cos(2\pi i/n),r_n\sin(2\pi i/n),1\right)^{\mathrm{T}}\in\mathbb{R}^3$ with $r_n:=\sqrt{\sec(\pi/n)}$. The set $\mathcal{E}(n)$ is the convex hull of the zero effect  $z\coloneqq (0,0,0)^{\mathrm{T}}$, the unit effect $u\coloneqq (0,0,1)^{\mathrm{T}}$ and the extremal effects $\{e_i,\Bar{e_i}=u-e_i\}_{i=1}^{n}$ where,
\begin{table}[h!]
\centering
\begin{tabular}{c||c}
\hline
~~~~~~~~Odd gon~~~~~~~~ & ~~~~~~~~Even gon~~~~~~~~ \\
\hline\hline
&\\
$e_i:=\frac{1}{1+r_n^2}\begin{pmatrix}r_n\cos\frac{2\pi i}{n}\\r_n\sin\frac{2\pi i}{n}\\1\end{pmatrix}$  & $e_i:=\frac{1}{2}\begin{pmatrix}r_n\cos\frac{(2i-1)\pi }{n}\\r_n\sin\frac{(2i-1)\pi }{n}\\1\end{pmatrix}$\\
\hline
\end{tabular}
\end{table}
Although these are just toy models, in the recent past several interesting results have been reported exploring different non-classical aspects of these models which in turn provide better understandings of the Hilbert space structure of quantum theory \cite{Massar14,Muller12,Safi15,Banik19,Saha20,Bhattacharya20}. Here we study these models to explore their communication utility while playing the Restaurant games.  

\subsubsection{Playing $\mathbb{H}^n(1/n)$ with Polygons} 
For a given $n$, Alice choose the system $\mathcal{P}_{ly}(2n)$ and sends the state $s_i=\frac{1}{2}\left(\omega_{2i-1}+\omega_{2i}\right)$ to Bob when $i^{th}$ Restaurant is closed. Note that, unlike the quantum case, here mixed states may be used for encoding. However, the states lies at the boundary of the polygons. For decoding, Bob performs an $n$-outcome measurements $\mathcal{M}(n)\equiv\left\{\frac{2}{n}\Bar{e}_{2k}\right\}_{k=1}^{n}$ and visits $k^{th}$ Restaurant if outcome corresponding to the effect $\frac{2}{n}\Bar{e}_{2k}$ is observed. The fact, $\Bar{e}_{2i}\cdot \omega_{2i-1} =0$ and $\Bar{e}_{2i}\cdot \omega_{2i} =0$ yield $p(i_b|i_c)= \frac{2}{n}\Bar{e}_{2i}\cdot s_i =0$. Furthermore, $p(m_b) = \frac{1}{n}\sum_{k=1}^n p(m_b|k_c)= \frac{1}{n}\sum_{k=1}^n \frac{2}{n}\Bar{e}_{2m}\cdot s_k= \frac{2}{n^2}\Bar{e}_{2m}\cdot \sum_{k=1}^n s_k  =\frac{2}{n^2}\Bar{e}_{2m}\cdot nu = \frac{1}{n}$. Thus both the conditions ($\mathrm{h}1$) and ($\mathrm{h}2$) are satisfied (see Fig.\ref{fig7}). 
\begin{figure}[b!]
\centering
\includegraphics[width=0.415\textwidth]{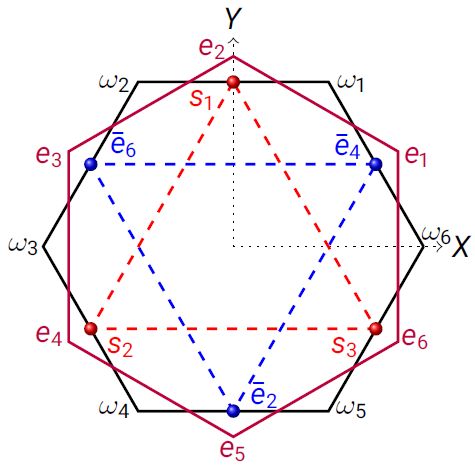}
\caption{(Color online) Strategy for the  $\mathbb{H}^3(1/3)$ game with communication of a single $\mathcal{P}_{ly}(6)$ system from Alice to Bob (without shared randomness). Red dots denote Alice's encodings. Blue dots are the effects proportional to Bob's decoding measurement.}\label{fig7}
\end{figure}
\begin{figure}[t!]
\centering
\includegraphics[width=0.4\textwidth]{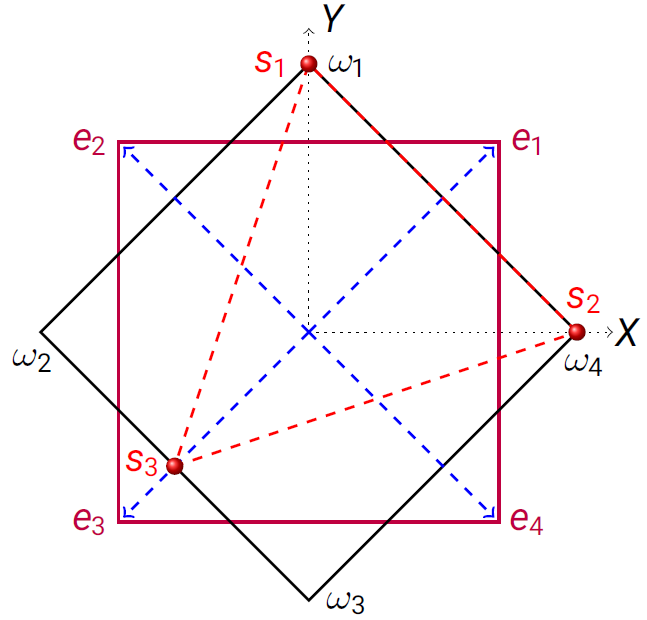}
\caption{(Color online) Strategy for the  $\mathbb{H}^3(1/3)$ game with communication of single $\mathcal{P}_{ly}(4)$ system from Alice to Bob (without shared randomness). Red dots denote Alice's encodings. Blue dashed lines are the effects proportional to Bob's decoding measurement.}\label{fig8}
\end{figure}

Note that, so far all the advantageous strategies discussed with qubit and polygonal models does not invoke any local randomness at the decoding step by Bob. However, using local randomness at Bob's end it can be further argued that the game $\mathbb{H}^3(1/3)$ can also be perfectly won with $\mathcal{P}_{ly}(4)$ system. For that, Alice uses the encoding $\mathbb{E}\equiv\left\{1\mapsto\omega_1,~2\mapsto\omega_4,~3\mapsto s:=\frac{1}{2}(\omega_2 +\omega_3 )\right\}$. For decoding, Bob performs the measurement $\mathcal{M}\equiv\left\{\frac{1}{2}e_1,\frac{1}{2}e_2,\frac{1}{2}e_3,\frac{1}{2}e_4\right\}$ which is a convex mixture (use of local randomness at decoding step) of the measurements $\mathcal{M}_1\equiv\{e_1,e_3\}$ and $\mathcal{M}_2\equiv\{e_2,e_4\}$ (see Fig.\ref{fig8}). Bob visits the Restaurant $3,~2$ and $1$, respectively when the outcome corresponding to the effect $e_1/2,~e_2/2$ and $e_4/2$ clicks. For the effect $e_3/2$ he visits Restaurant $1$ and $2$ with equal probability. It is a straightforward calculation to check that both the conditions ($\mathrm{h}1$) and ($\mathrm{h}2)$ are satisfied with this strategy. In fact, it turns out that any $\mathbb{H}^3(\gamma_1.\gamma_2,\gamma_3)$ game can be perfectly won with the $\mathcal{P}_{ly}(4)$ system (see Appendix \ref{gen-3-Restaurant with square}).

At this point, the work by Massar \& Patra is worth mentioning \cite{Massar14}. They have shown that the classical capacity (in Holevo sense {\it i.e.}, in the asymptotic limit) of $n$-gon system is $1$ bit for even $n$,  whereas it is larger than $1$ bit for odd $n$. In other words, all odd-gons have more communication utility than a single classical bit. On the other hand, the authors in \cite{Czekaj17} have reported the super-quantum communication utility of polygon models within the W-ANTV communication scenario. More recently, the super-quantum communication utility of $n=4$ system (square bit) has been established within H-FW scenario by considering suitable bipartite composition of this system \cite{DallArno17}. Our Restaurant games, however, establish super-classical communication utility of all even gon models in H-FW scenario from single copy consideration only. It might be interesting to compare the communication utility of a qubit with polygon models while playing the Restaurant games. We defer this issue to the next section.

\section{The more shared randomness there is, the greater the utility.}
So far we have seen that the games $\mathbb{H}^n(\gamma_1,\cdots,\gamma_n)$ proposed above have several novel implications. On  one hand, they establish the important role of classical shared randomness in enhancing the communication utility of a perfect classical channel, on the other hand, they also establish the superiority of qubit (as well as polygon systems) communication over a $1$-bit classical channel in H-FW communication scenario. In Section \ref{1bit1SR} it has also been shown that an arbitrary  $\mathbb{H}^3(\gamma_1,\gamma_2,\gamma_3)$ game can always be perfectly won when $1$-bit of classical communication is assisted with $1$-bit of classical shared randomness. Here we propose a class of games to show that classical communication lines may need more assistance of shared randomness for perfect winning.    

{\bf Strict Restaurant game:} The game denoted as $\mathbb{H}^n[1/(n-1)]$\footnote{Please note that here we use `square bracket' to denote the stricter version game $\mathbb{H}^n[\star]$ to make the distinction from its non-stricter version.} can be seen as a  stricter version of $\mathbb{H}^n(1/n)$. Here the conditions ($\mathrm{h}1$) and ($\mathrm{h}2$) are modified as follows
\begin{itemize}
\item[($\mathrm{h}_s1$)] Bob never visits a closed Restaurant, {\it i.e.} $p(i_b|i_c)=0,~\forall~i\in\{1,\cdots,n\}$.
\item[($\mathrm{h}_s2$)] Whenever a certain Restaurant is closed all other Restaurants must be visited with equal probability, {\it i.e.} $p(i_b|j_c\neq i)=\frac{1}{n-1},~\forall~i,j\in\{1,\cdots,n\}$.   
\end{itemize}
The visit matrix in Eq.(\ref{visit}) with this stricter winning conditions can be written compactly as 
\begin{align}\label{strict win}
p(i_b|j_c)=\frac{1}{n-1}(1-\delta_{ij}), \end{align}
where $\delta_{ij}$ denotes the Kronecker delta. Note that, the condition ($\mathrm{h}_s1$) is exactly same as the condition ($\mathrm{h}1$). On the other hand, condition $(\mathrm{h}_s2)$ always implies condition ($\mathrm{h}2)$, but the converse is not the case in general, which makes this game strict.

\subsection{Unwinnability of $\mathbb{H}^4[1/3]$ with $1$-bit communication $+1$-bit shared randomness}\label{msrmu}
For the game $\mathbb{H}^4[1/3]$ the winning conditions read as,
\begin{align}
\label{strict win_3}
p(i_b|j_c)=\frac{1}{3}(1-\delta_{ij});~~~i,j\in\{1,\cdots,4\}.
\end{align}
As a physical source of classical shared randomness we can consider that Alice and Bob share a two-qubit classically correlated state $\rho_{AB}=\lambda\ket{00}_{AB}\bra{00}+(1-\lambda)\ket{11}_{AB}\bra{11}$, with $\lambda\in[0,1]$. Outcomes $s\in\{0,1\}$ of their $\sigma_z$ measurement are correlated in this state. They have $1$-bit of shared randomness whenever $\lambda=1/2$, whereas $\lambda\in\{0,1\}$ implies no shared randomness.     

In assistant with $1$-bit of shared randomness, the most general strategy that Alice and Bob can follow is a convex mixture of two mixed strategies that can be implemented as the following steps.
\begin{itemize}
\item[(S-1)] Depending on the outcome $s\in\{0,1\}$ of the measurement $\sigma_z$ on her part of the state $\rho_{AB}$ and based on the information about the closed Restaurant $k\in\{1,\cdots,4\}$, Alice tosses a $2$ sided biased coin having the outcomes $\{0,1\}$. The outcome probabilities of the coins are given by $P_k^{(s)}(0)=\alpha_k^{(s)}$ and $P_k^{(s)}(1)=1-\alpha_k^{(s)}$.
\item[(S-2)] Alice communicates the outcome of her coin toss to Bob through the $1$-bit classical channel. 
\item[(S-3)] Depending on the outcome $s\in\{0,1\}$ of $\sigma_z$ measurement on his part of the state $\rho_{AB}$, Bob prepares two $4$ sided coins with outcomes $\{1,\cdots,4\}$ having the outcome probabilities $\vec{r}^{(s)}=(r_1^{(s)},\cdots,r_4^{(s)})$ and $\vec{q}^{(s)}=(q_1^{(s)},\cdots,q_4^{(s)})$, respectively. Upon receiving $0$ from Alice, he tosses the $\vec{r}^{(s)}$ coin and visits the $i^{th}$ Restaurant if $i^{th}$ outcome occurs. He follows a similar strategy with the $\vec{q}^{(s)}$ coin if $1$ is received from Alice.
\end{itemize}
With this strategy, the conditional probability $p(m_b|k_c)$ that Bob visits $m^{th}$ Restaurant provided the $k^{th}$ Restaurant is closed, turns out to be,
\begin{align}
\label{eq18}
p(m_b|k_c)=&\lambda\left[\alpha_k^{(0)}\times r_m^{(0)}+\left(1-\alpha_k^{(0)}\right)\times q_m^{(0)}\right]+\nonumber\\
(1-&\lambda)\left[\alpha_k^{(1)}\times r_m^{(1)}+\left(1-\alpha_k^{(1)}\right)\times q_m^{(1)}\right].
\end{align}
The first condition ($\mathrm{h}_s1$) demands $p(i_b|i_c)=0,~\forall~i\in\{1,\cdots,4\}$, which consequently implies,
\begin{align}
\label{eq19}
&\forall~i,~~\lambda\left[\alpha_i^{(0)}\times r_i^{(0)}+\left(1-\alpha_i^{(0)}\right)\times q_i^{(0)}\right]+\nonumber\\
&(1-\lambda)\left[\alpha_i^{(1)}\times r_i^{(1)}+\left(1-\alpha_i^{(1)}\right)\times q_i^{(1)}\right]=0.
\end{align}
It is easy to argue that in absence of shared randomness, {\it i.e.} when $\lambda\in\{0,1\}$, the goal is impossible to achieve. For $\lambda=0$ we have $\alpha_i^{(1)}\times r_i^{(1)}+\left(1-\alpha_i^{(1)}\right)\times q_i^{(1)}=0~\forall~i~\in~ \{1,\cdots 4\}$. This boils down to coming up with a partition $X$ and $Y$ as discussed in Section \ref{Games winnable with mixed strategies}. However, this scenario can be disregarded immediately. As there are $4$ Restaurants, no matter which partition $X$ and $Y$ of the Restaurants we take, there must exist at-least $2$ Restaurants in either $X$ or $Y$. Now if two Restaurants $m$ and $k$ belong to the same partition we must have $p(m_b|k_c)=0$; and hence violates the winning condition ($\mathrm{h}_s2$). A similar argument holds for the case $\lambda=1$. We now go on to show that even with the assistance of shared randomness, {\it i.e.} for $\lambda\notin\{0,1\}$, the goal is impossible to achieve. Since $\lambda\in(0,1)$, to satisfy Eq.(\ref{eq19}) we must have
\begin{align*}
\alpha_i^{(0)}\times r_i^{(0)}&=0,~~~\left(1-\alpha_i^{(0)}\right)\times q_i^{(0)}=0,\\
\alpha_i^{(1)}\times r_i^{(1)}&=0,~~~\left(1-\alpha_i^{(1)}\right)\times q_i^{(1)}=0,\\
&~~~~~\forall~i~\in\{1,2,\cdots 4\}.~\nonumber
\end{align*}
All the four equations contain product of two terms on the left hand side, and for each pair at least one must be zero. This in turn leads to several possible ways to satisfy the equations. The case where left term of a particular equation is $0$ will be denoted by `$0$', whereas `$1$' will indicate the right term is $0$. So each Restaurant must be assigned a four bit string indicating which of the terms in these four equations are $0$. For instance, lets say we assign a string `$0110$' to the $i^{th}$ Restaurant. The first bit of this string implies $\alpha_i^{(0)}=0$, second implies $q_i^{(0)}=0$, and similarly third and fourth imply  $r_i^{(1)}=\left(1-\alpha_i^{(1)}\right)=0$. It becomes immediate that the possibilities corresponding to the set of strings $\{0000,0001,0010,0011,0100,1000,1100\}$ are not allowed for any $i\in\{1,\cdots,4\}$. For all these cases we have $\alpha_i^{(s)}=1-\alpha_i^{(s)}=0$ for at least one $s\in\{0,1\}$, which is impossible. Furthermore, the possibility corresponding to the string `$1111$' is also not allowed for any $i$. In this case, although $p(i_b|i_c)=0$ and hence the condition ($\mathrm{h}_s1$) is satisfied, it turns out that $p(i_b|j_c\neq i)=0$ violating the  condition ($\mathrm{h}_s2$). Thus for any of the Restaurants, the allowed possibilities are   
\begin{align*}
\left\{\!\begin{aligned}
S_1&=0111,~~S_2=1011,~~S_3=1101,~~S_4=1110,\\
S_5&=0101,~~S_6=0110,~~S_7=1001,~~S_8=1010
\end{aligned}\right\}.
\end{align*}
\begin{center}
\begin{table}[t!]
\setlength
\extrarowheight{5pt}
\begin{tabular}{ c||c }
\hline
String~($S_i$) & ~~~~~~~~~~Strings compatible with $S_i$~~~~~~~~~~\\
\hline\hline
$S_1$ & $S_2,~S_3,~S_4,~S_7,~S_8$ \\
\hline
$S_2$ & $S_1,~S_3,~S_4,~S_5,~S_6$\\
\hline
$S_3$ & $S_1,~S_2,~S_4,~S_6,~S_8$\\
\hline
$S_4$ & $S_1,~S_2,~S_3,~S_5,~S_7$\\
\hline
$S_5$ & ~~~$S_1,~S_2,~S_3,~S_4,~S_6,~S_7,~S_8$~~~\\
\hline
$S_6$ & ~~~$S_1,~S_2,~S_3,~S_4,~S_5,~S_7,~S_8$~~~\\
\hline
$S_7$ & ~~~$S_1,~S_2,~S_3,~S_4,~S_5,~S_6~,~S_8$~~~\\
\hline
$S_8$ & ~~~$S_1,~S_2,~S_3,~S_4,~S_5,~S_6,~S_7$~~~\\
\hline
\end{tabular}
\caption{On the right column we list all the strings that are compatible with a given string $S_k$; $k\in\{1,\cdots,8\}$. Although the string $S_1$ is compatible with the string $S_5$ (see $5^{th}$ row), $S_5$ is not compatible with $S_1$ (see $1^{st}$ row).}
\label{table4}
\end{table}
\end{center}
\vspace{-.4cm}
We will use the notation $i^{(S_d)}$ to denote that $i^{th}$ Restaurant is assigned the string $S_d$. Let the string $S_1=0111$ is assigned to the $i^{th}$ Restaurant, then using Eq.(\ref{eq18}) we obtain $p(i_b^{(S_1)}|j_c)=\lambda[\alpha_j^{(0)}\times r_i^{(0)}]$. If the same string ($S_1$) is assigned to the $j^{th}$ Restaurant for $j\neq i$, then we must have $\alpha_j^{(0)}=0$. Consequently we have $p(i_b^{(S_1)}|j_c^{(S_1)})=0$, which violates the condition ($\mathrm{h}_s2$). Thus two different Restaurants cannot be assigned the same string, confining us to select $4$ different strings for $4$ different Restaurants. It is important to note that for $i\neq j$,  $p(i_b^{(S_5)}|j_c^{(S_1)})$ need not be $0$, although $p(i_b^{(S_1)}|j_c^{(S_5)})=0$. We say that $S_1$ is compatible with $S_5$, but $S_5$ is incompatible with $S_1$ as it violates ($\mathrm{h}_s2$). The list of compatible assignments of the strings are listed in Table \ref{table4}. To satisfy the condition $p(m_b|k_c)=1/3~\forall~(m\neq k)$ we must select $4$ strings that are compatible with each other. From Table \ref{table4} we only have two possible ways of selecting the strings: (i) $\equiv\{S_1,S_2,S_3,S_4\}$ or (ii) $\equiv\{S_5,S_6,S_7,S_8\}$. We are left to show that none of these two cases yields any consistent solutions for all the variables $\left\{\lambda,\alpha_k^{(s)},\vec{r}^{(s)},\vec{q}^{(s)}\right\}$. 

(i) $\{S_1,S_2,S_3,S_4\}$: Without loss of any generality we can choose the string $S_i$ for $i^{th}$ Restaurant. Then using Eq.(\ref{strict win_3}) \&  Eq.(\ref{eq18}) we obtain
\begin{align*}
p(1|2)&=\lambda r_1^{(0)}=\frac{1}{3},~~~p(1|3)=\lambda \alpha_3^{(0)}r_1^{(0)}=\frac{1}{3},\\
p(2|1)&=\lambda q_2^{(0)}=\frac{1}{3},~~~p(2|3)=\lambda(1-\alpha_3^{(0)})q_2^{(0)}=\frac{1}{3}.
\end{align*}
A consistent solution does not exist for these equations.

(ii) $\{S_5,S_6,S_7,S_8\}$: Once again without loss of an generality we can choose the string $S_{i+4}$ for $i^{th}$ Restaurant. Once agin Eq.(\ref{strict win_3}) \& Eq.(\ref{eq18}) yield
\begin{align*}
p(1|2)&=(1-\lambda) r_1^{(1)}=\frac{1}{3},~~~p(1|3)=\lambda r_1^{(0)}=\frac{1}{3},\\
p(1|4)&=\lambda r_1^{(0)}+(1-\lambda) r_1^{(1)}=\frac{1}{3}.
\end{align*}
No consistent solution exists in this case too. Thus the game $\mathbb{H}^4[1/3]$ is not possible to win with $1$-bit of classical communication even when the communication channel is assisted with $1$-bit of classical shared randomness. Naturally, the question arises whether  this game is winnable with $1$-bit of classical communication if more shared randomness is supplied. 
\begin{theorem}\label{theo3}
In the H-FW scenario the communication utility of 1 cbit channel is increased as the amount of  assisting shared randomness is increased from 1 bit to $\log 3$ bit.
\end{theorem}
\begin{proof}
Although the game $\mathbb{H}^4[1/3]$ cannot be won with $1$-bit+$1$-SR resource, here we will show that the goal can be perfectly achieved through a $1$-bit classical channel if more shared randomness are provided as assistance. The strategy is described below. 

Alice divides the Restaurants into two disjoint partitions $X$ and $Y$, such that $ |X|=|Y|=2$. Alice sends $0$ to Bob whenever a Restaurant in $X$ is closed, otherwise, she sends $1$. Bob visits the Restaurants within the set $Y$ ($X$) with uniform probability if he receives $0~(1)$ from Alice. Clearly, the condition ($\mathrm{h}_s1$) is satisfied as Bob never visits a closed Restaurant. Consider the following three partitionings $X$ and $Y$.
\begin{itemize}
\item[C-$1$:] $X=\{1,2\}~\&~Y=\{3,4\}$,
\item[C-$2$:] $X=\{1,3\}~\&~Y=\{2,4\}$,
\item[C-$3$:] $X=\{1,4\}~\&~Y=\{2,3\}$.
\end{itemize}
These correspond to three different strategies leading to three different visit matrices 
\footnotesize
\begin{align*}
\mathbb{V}_1= \left(\begin{array}{cccc}
     0  &0&\frac{1}{2}&\frac{1}{2} \\
    0&0&    \frac{1}{2}&\frac{1}{2}\\
      \frac{1}{2}&\frac{1}{2}&0&0\\
         \frac{1}{2}&\frac{1}{2} & 0&0
  \end{array}\right),
\mathbb{V}_2= \left(\begin{array}{cccc}
      0&\frac{1}{2}&0&\frac{1}{2} \\
       \frac{1}{2}&0&\frac{1}{2}&0\\
     0&\frac{1}{2}&0&\frac{1}{2}\\
         \frac{1}{2}&0&\frac{1}{2} & 0
  \end{array}\right),
\mathbb{V}_3=\left(\begin{array}{cccc}
     0  & \frac{1}{2}&\frac{1}{2}&0 \\
        \frac{1}{2}&0&0&\frac{1}{2}\\
     \frac{1}{2}&0&0&\frac{1}{2}\\
         0&\frac{1}{2}&\frac{1}{2} & 0
  \end{array}\right).
\end{align*}
\normalsize
An equal mixture of these three strategies, which requires Alice and Bob to share $\log3$-bit of shared randomness, yields the resulting visit matrix 
\begin{align*}
\mathbb{V}=\left(\begin{array}{cccc}
0&1/3&1/3&1/3\\
1/3&0&1/3&1/3\\
1/3&1/3&0&1/3\\
1/3&1/3&1/3&0
\end{array}\right),
\end{align*}
which satisfies both the conditions ($\mathrm{h}_s1$) and ($\mathrm{h}_s2$).

$\mathbf{\mathbb{H}^n[1/(n-1)]}:$ The above protocol can be generalized for  $\mathbb{H}^{n}[1/(n-1)]$ game with $1$-bit of classical communication assisted with $\log (n-1)$-bit of SR. Alice sends $0$ and $1$ to direct Bob to visit $k$-th and $(k+1)$-th Restaurant respectively, where $k\in\{1,2,\cdots,(n-1)\}$. The value of the $k$ will be identified by the outcomes $\{1,2,\cdots,(n-1)\}$ of the SR, so it requires $\log (n-1)$-bit of SR. Whenever the $m$-th Restaurant ($m\in\{2,3,\cdots,(n-1)\}$) is closed, Alice communicates $0$ and $1$ respectively for every $k\in\{1,\cdots,(m-1)\}$-th and $k\in\{m,\cdots,(n-1)\}$-th outcomes of the SR. On the other hand, if the $1$-st or, $n$-th Restaurant is closed, then Alice will communicate $0$ and $1$ respectively, independent of the SR outcomes. It is easy to verify that the strategy satisfies both the required conditions. 
 \end{proof}
\begin{remark}
At this point it is important to note that a $\mathbb{H}^{n}[1/(n-1)]$ game cannot be won with a qubit communication alone, whenever $n\ge5$. This establishes the `order of merit' Q$\prec_{inst}$C+SR as listed in Table \ref{table5}.  
\end{remark}

\subsection{Perfect winnability of $\mathbb{H}^4[1/3]$ with qubit strategy}\label{strictq}
Although the game $\mathbb{H}^4[1/3]$ cannot be won with $1$-bit communication $+1$-bit shared randomness, the goal can be achieved through a qubit communication without any assistance of classical shared randomness. Alice sends the qubit state $\ket{\psi_i}$ when the $i^{th}$ Restaurant is closed. The Bloch vectors corresponding the encoded states are respectively 
\begin{figure}[t!]
\centering
\includegraphics[width=0.4\textwidth]{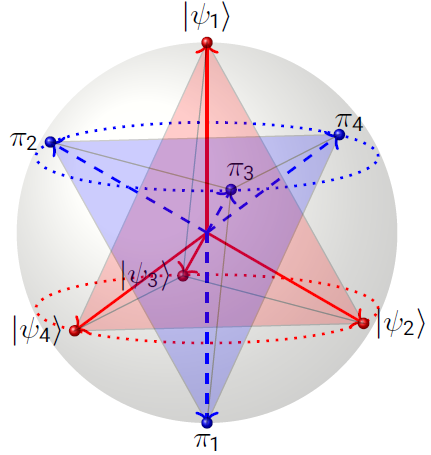}
\caption{(Color online) Alice's encodings (red dots) form a symmetric tetrahedron inscribed within the Bloch sphere. For decoding Bob performs SIC-POVM corresponding to the inverted tetrahedron (blue dots) of Alice's encoding tetrahedron.}\label{fig9}
\end{figure}
\begin{align*}
1&\mapsto \ket{\psi_1}\equiv(0,0,1)^{\mathrm{T}},\\
2&\mapsto \ket{\psi_2}\equiv\left(\frac{2\sqrt{2}}{3},0,\frac{-1}{3}\right)^{\mathrm{T}},\\
3&\mapsto \ket{\psi_3}\equiv\left(\frac{2\sqrt{2}}{3}\cos\frac{2\pi}{3},\frac{2\sqrt{2}}{3}\sin\frac{2\pi}{3},\frac{-1}{3}\right)^{\mathrm{T}},\\
4&\mapsto \ket{\psi_4}\equiv\left(\frac{2\sqrt{2}}{3}\cos\frac{4\pi}{3},\frac{2\sqrt{2}}{3}\sin\frac{4\pi}{3},\frac{-1}{3}\right)^{\mathrm{T}}.
\end{align*}
The encodings form a symmetric tetrahedron within the Bloch sphere (see Fig.\ref{fig9}). For decoding, Bob performs the measurement $\mathcal{M}(4)\equiv\left\{\pi_k:=\frac{1}{2}\ketbra{\psi_k^{\perp}}{\psi_k^{\perp}}\right\}_{k=1}^4$ and visits the $k^{th}$ Restaurant if outcome corresponding to the effect $\pi_k$ is observed. Note that, the decoding corresponds to the SIC-POVM described by the tetrahedron inverted to the encoding tetrahedron. The fact, $\tr\left(\ket{\psi_i}\bra{\psi_i}\pi_k\right)=\frac{1}{3}\left(1-\delta_{ik}\right)$ ensures perfect winning of the game $\mathbb{H}^4[1/3]$.  For the comparison of resources we will use $\mathcal{R}_1\prec_{inst}\mathcal{R}_2$ to imply that there is some instance where $\mathcal{R}_2$ is strictly better than $\mathcal{R}_1$, while in a different instance their utilities might be in reverse order. So by considering the $\mathbb{H}^4[1/3]$ game we conclude the following theorem.
\begin{theorem}
   There exist a task in the H-FW scenario where the  utility of 1 qubit is higher than the utility of 1 cbit with the assistance of 1 bit of shared randomness i.e. $C+ 1 SR \prec_{inst}Q$.
\end{theorem}

\subsection{Unwinnability of $\mathbb{H}^4[1/3]$ with Polygons}\label{qbtp}
To ensure that condition ($\mathrm{h}_s1$) with a polygon system, Alice's must choose her encodings $\{\omega_i\}_{i=1}^4$ from the boundary of the polygon. Otherwise, no effect in a polygon model will yield zero probability for a given encoding state. For an $e_i$ if we move towards the effect $\tilde{e_i}$, with $e_i$ and $\tilde{e_i}$ forming a measurement, the probability of the probability of observing the effect ${e_i}$ decreases gradually. For the game $\mathbb{H}^4[1/3]$, the  decoding strategy of Bob must be a $4$ outcome measurement $\mathcal{M}(4)\equiv\left\{\left(\alpha_i,e_i\right)~|~\sum_{i=1}^4\alpha_ie_i=u\right\}$ such that Bob visits $i^{th}$ Restaurant when the outcome corresponding to the effect $e_i$ is observed. For Alice's encoding $\{\omega_j\}_{j=1}^4$ the winning condition of $\mathbb{H}^4[1/3]$ demands $\alpha_ie_i^\mathrm{T}\omega_j=\left(1-\delta_{ij}\right)/3$
\begin{figure}[b!]
\centering
\includegraphics[width=0.4\textwidth]{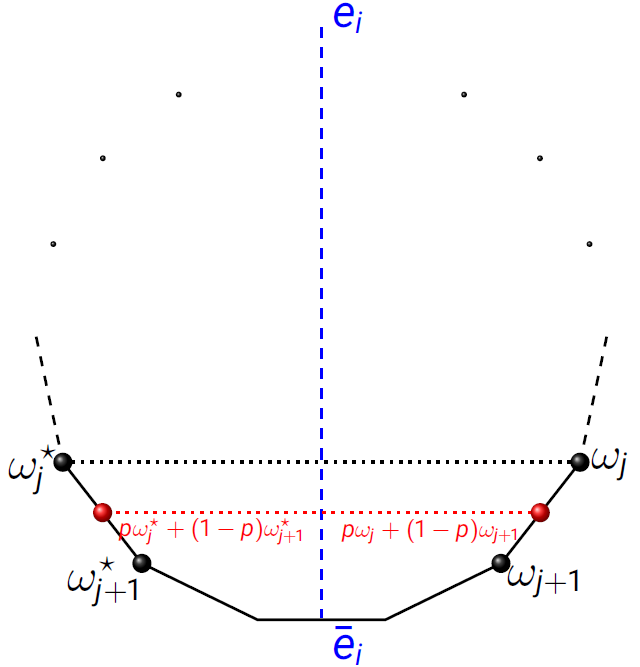}
\caption{For any polygon theory the probability of observing an effect $e_i$ on a state $\omega_j$ keep decreasing as the state moves from $e_i$ towards $\tilde{e}_i$. Since $e_i$ and $\bar{e}_i$ form a  measurement, {\it i.e.}, $e_i+\bar{e}_i=u$, the mirror image $\omega_j^*$ of any state $\omega_j$ about the line passing through $e_i$ and $\tilde{e}_i$ will give the same probability of occurrence the effect $e_i$.}\label{fig10}
\end{figure}

For any effect $e_i$ there exists an state $\omega_j$ such that $e_i^{\mathrm{T}}\omega_j\leq\frac{1}{3}<e_i^{\mathrm{T}}\omega_{j+1}$. Let the first encoding state is $p\omega_j+(1-p)\omega_{j+1}$, where $p\in[0,1)$. Symmetry of the state spaces imply that probability of observing $e_i$ on any state $\omega_j$ is same as that on the state $\omega_j^*$, where $\omega_j^*$ is the mirror image of $\omega_j$ about the line joining $e_i$ and $\tilde{e}_i$, where $e_i+\bar{e}_i=u$ (see Fig. \ref{fig10}). Therefore, the second encoding state must be $p\omega_j^*+(1-p)\omega_{j+1}^*$, which is the mirror image of the first encoding state about the line joining $e_i$ and $\bar{e}_i$. As the probability of observing the effect $e_i$ keeps decreasing gradually as the choices of the encoding state  move away from $e_i$ to $\bar{e_i}$, there will not be a third encoding state yielding the winning probability for the effect $e_i$. Thus, a set of $4$ distinct encoding states does not exist that satisfy the winning condition. So as a comparison of different communication resources we can establish the following theorem.
\begin{theorem}
    There exist a task in the H-FW scenario where in the absence of any preshared correlation the utility of 1 qubit is greater than 1 polygon bit i.e. Polygon$\prec_{inst}$ Q.
\end{theorem}

\section{Order of merit}\label{OoM}
Present work along with some recent works put interesting `order of merit' among different communication resources to compare their utility in the H-FW communication scenario. Following three different ordering relations we will use.
\begin{itemize}
\item A resource $\mathcal{R}_1$ will be same as a different resource $\mathcal{R}_2$, denoted as $\mathcal{R}_1=\mathcal{R}_2$ if for any communication task within H-FW scenario they provide same utility.
\item We will denote $\mathcal{R}_1\preccurlyeq\mathcal{R}_2$ if  $\mathcal{R}_2$ is as good as $\mathcal{R}_1$ for any such task and there exist at-least one task where $\mathcal{R}_2$ yield a strictly greater utility than $\mathcal{R}_1$.
\item Finally, $\mathcal{R}_1\prec_{inst}\mathcal{R}_2$ will imply that there is some instance where $\mathcal{R}_2$ is strictly better than $\mathcal{R}_1$, while in a different instance their utilities might be in reverse order.
\end{itemize}
Few important orderings are listed in Table \ref{table5}.
\begin{table}[t!]
\setlength
\extrarowheight{5pt}
\centering
\begin{tabular}{ |c||c|}
\hline
~~{\bf Order of merit}~~ & {\bf Task}\\
\hline\hline
Q $\preccurlyeq$ Q + $1$ ebit & \href{https://doi.org/10.1103/PhysRevLett.69.2881}{PRL {\bf 69}, 2881 (1992)}\\
\hline
C $\preccurlyeq$ C + $1$ ebit & \href{https://doi.org/10.22331/q-2022-03-01-662
}{Quantum {\bf 6}, 662 (2022)}\\
\hline
C+SR=Q+SR &\href{https://doi.org/10.1007/s00220-015-2463-0}{CMP {\bf 340}, 563 (2015)} \\
\hline
C $\preccurlyeq$ C+SR &  $\mathbb{H}^n(\gamma_1,\cdots,\gamma_n):$ Section \ref{1bit1SR}\\
\hline
C$\preccurlyeq$ Q & $\mathbb{H}^3(\gamma_1,\gamma_2,\gamma_3):$ Section \ref{quantum strategy}\\
\hline
C$\preccurlyeq$ $2n$-gon & $\mathbb{H}^n(1/n):$ Section \ref{polygon}\\
\hline
C+ $1$ SR$\prec_{inst}$Q & $\mathbb{H}^4[1/3]:$ Section \ref{strictq}\\
\hline
C+$1$ SR$\preccurlyeq$ C+$\log3$-SR& $\mathbb{H}^4[1/3]:$ Section \ref{msrmu}\\
\hline
Polygon$\prec_{inst}$ Q& $\mathbb{H}^4[1/3]:$ Section \ref{qbtp}\\
\hline
Q$\prec_{inst}$C+SR & $\mathbb{H}^n[1/(n-1)],~n\ge5:$ Appendix \ref{ccsfor strict}\\
\hline
Q+SR$\prec_{inst}$C+$1$ ebit& \href{https://doi.org/10.22331/q-2022-03-01-662}{Quantum {\bf 6}, 662 (2022)}\\
\hline
\end{tabular}
\caption{Utility of different resources in H-FW scenario. Here Q denotes a perfect qubit channel and C is a perfect $1$-bit classical channel. When we write just `SR' it means unbounded amount of this resource is allowed.}
\label{table5}
\end{table}

\section{Making the quantum advantage noise-robust}\label{noise}
Noise is an inevitable enemy of any communication line. In quantum scenario this is even more prominent as a little bit of thermal noise can destroy the encodings. Furthermore, due to measurement impression it might not be possible to perform the intended decoding measurement exactly. So it is quite important to investigate whether the quantum advantage reported above persists under a noisy circumstance. Here we mainly analyse the $\mathbb{H}^3(1/3)$ game. A similar analysis is possible for other cases too. \begin{figure}[b!]
\centering
\includegraphics[width=0.45\textwidth]{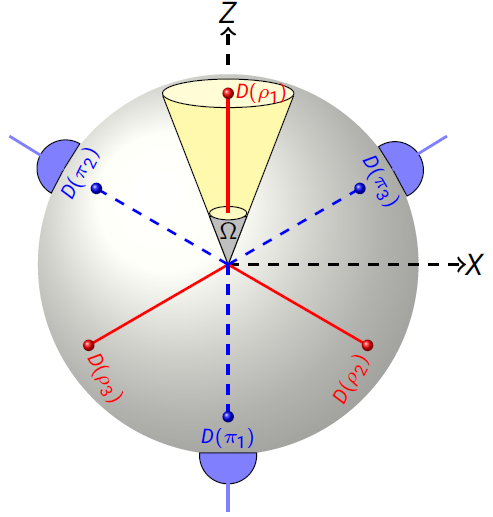}
\caption{(Color online) Noise robust quantum protocol for the game $\mathbb{H}^3(1/3)$. Under a depolarizing noise the encoded states get shrieked within the Bloch sphere. Both the encoded states and decoding effects are assumed to undergo such noise.}
\label{fig11}
\end{figure}
Note that any kind of noise that affects the entire Bloch sphere symmetrically, known as depolarizing ($D$) noise, does not change the visiting probability $p(m_b)$. Experimentally such a noise can arise when Alice is not able to prepare the ideal quantum preparations $\rho_k=\frac{1}{2}\left[\mathbb{I}+\hat{n}_k\cdot \vec{\sigma}\right]$, rather $\hat{n}_k$ takes values uniformly from a cone in the Bloch sphere that subtends solid angle $\Omega$ at the origin with axis $\hat{n}_k$ (see Fig.\ref{fig11}). For such a noise $D_\epsilon$, specified by depolarizing parameter $\epsilon\in[0,1]$, a state $\rho_k=\frac{1}{2}[\mathbb{I}+\hat{n}_k.\vec{\sigma}]$ gets modified as
\begin{align}
D_\epsilon(\rho_k)&:=(1-\epsilon)\rho_k+\epsilon\frac{\mathbb{I}}{2}=\frac{1}{2}[\mathbb{I}+(1-\epsilon)\hat{n}_k.\vec{\sigma}] \nonumber\\
&=\frac{1}{2}[\mathbb{I}+D_\epsilon(\hat{n}_k).\vec{\sigma}],~~D_\epsilon(\hat{n}_k):=(1-\epsilon)\hat{n}_k.
\end{align}
Similarly, the POVM elements $\pi_k=\frac{1}{3}\left[\mathbb{I}-\hat{n}_k\cdot \vec{\sigma}\right]$ get modified to,
\begin{align}
D_\epsilon(\pi_k)&:=(1-\epsilon)\pi_k+\epsilon\frac{\mathbb{I}}{3}=\frac{1}{3}[\mathbb{I}+(1-\epsilon)\hat{n}_k.\vec{\sigma}] \nonumber\\
&=\frac{1}{3}[\mathbb{I}-D_\epsilon(\hat{n}_k).\vec{\sigma}],~~D_\epsilon(\hat{n}_k):=(1-\epsilon)\hat{n}_k.
\end{align}
Consider that both Alice's encoding and Bob's decoding get affected with such noises $D_{\epsilon_e}$ and $D_{\epsilon_d}$, respectively. For the game $\mathbb{H}^3(1/3)$, due to the symmetric nature of the noise, Bob's probability of visiting a Restaurant will remain the same as before, {\it i.e.} $p(m_b)=1/3$. However, there is a finite probability of visiting a closed Restaurant by Bob, with the probability depending on the noise parameters $\epsilon_s$ and $\epsilon_e$. The probability of visiting the $k^{th}$ when it is closed is,
\begin{align}
p(k_b|k_c)&=\Tr\left[D_{\epsilon_{s]e}}(\rho_k) D_{\epsilon_{d}}(\pi_k)\right]\nonumber\\
&=\frac{1}{3}(\epsilon_{e}+\epsilon_{d}-\epsilon_{e}\epsilon_{d}).
\end{align}
If a classical strategy without shared randomness is used, and we put the conditions that ($\mathrm{h}_n1$) a closed Restaurants can be visited with probability less than $1/6$ and ($\mathrm{h}_n2$) Each Restaurant must be visited with equal probability, then no solution exists; here the sub-index $n$ is used to denote that the perfect conditions is modified for noisy case. In other words, there is no classical mixed strategy without shared randomness that satisfies the conditions:
\begin{align}
p(m_b)=\frac{1}{3}~~\&~~p(m_b|m_c)< \frac{1}{6},~\forall m \in \{1,2,3\}.
\end{align}
Therefore, a noisy quantum strategy gives advantage over any classical mixed strategy whenever noise parameters $\epsilon_s$ and $\epsilon_e$ are small enough to obey the condition $(\epsilon_{e}+\epsilon_{d}-\epsilon_{e}\epsilon_{d}) <1/2$ as illustrated in Fig \ref{fig12}. 
\begin{figure}[t!]
\centering
\includegraphics[width=0.45\textwidth]{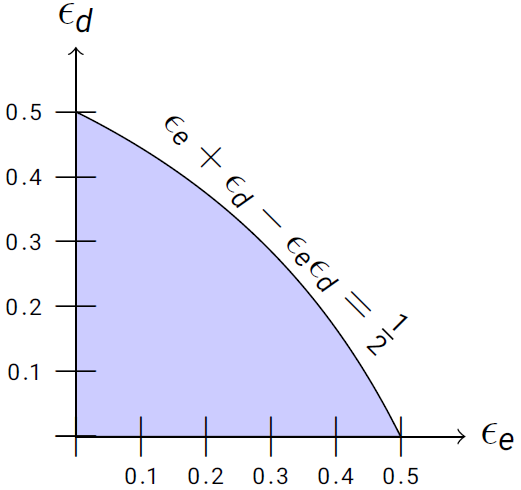}
\caption{Trade-off in the noise parameters $\epsilon_e$ and $\epsilon_d$ allowing quantum advantage in $\mathbb{H}^3(1/3)$ when both the encodings and decoding experience depolarizing noise.}
\label{fig12}
\end{figure}
At this point one may wish to analyze a different type of noise in quantum systems and carry out the analysis along the lines that we have done above. For such purposes it is convenient to define a function that captures the error obtained while trying to play our Restaurant game. One such function is given by
\begin{align}
\mathscr{E}(s)=\frac{1}{3}\sum_i\left[p^s(i_b|i_c)+(\gamma_i^s-1/3)^2\right];\label{error}
\end{align}
where $\gamma_i^s:=\frac{1}{3}\sum_{j=1}^3p^s(j_b|i_c)$ and `$s$' denotes the strategy followed by Alice and Bob. The collective aim of Alice and Bob is to come up with a strategy that optimize this error. For a perfect strategy ($s_p$) that satisfies both ($\mathrm{h}1$) and ($\mathrm{h}2$) we have $\mathscr{E}(s_p)=0$. A general mixed strategy is characterized by $7$ parameters -- $3$ for three different $2$-faced coins on Alice's part and $4$ for two $3$-faced coins on Bob's part. Using Monte Carlo simulation over this $7$-parameter space we find that for classical mixed strategies, $\mathscr{E}(s)$ is lower bounded by $0.108$. Thus any noisy quantum strategy, having the error value less than this value will establish the quantum advantage. 

\section{Advantage of Shared Randomness for worst-case success}\label{worst-case success}
The advantage of shared randomness in H-FW scenario can also be shown in a much simpler game. A referee provide Alice some classical random variable $x\in\{1,2,3\}$ and Bob generate some random variable $b\in\{1,2,3\}$. Bob has to guess Alice's input but  communication from Alice to Bob is restricted to $1$ cbit. Now the referee will consider the worst case success probability of Bob correctly guessing Alice's input, which is given by $P_{suc}^{W}=\min_{x}P(b=x|x)$.\\
{\bf Classical Strategy}: Let us consider particular encoding and decoding strategy for Alice and Bob. If Alice receive `$1$' as her input she communicate `$0$' to Bob, otherwise `$1$' is communicated. Whenever Bob receives `$0$' he answers `$1$' and upon receiving `$1$' half of the time he  answers `$2$' and other half of the time he will answers `$3$'. The success probability for different input $x$ using this strategy are $P(1|1)=1$, $P(2|2)=\frac{1}{2}$, $P(3|3)=\frac{1}{2}$. So the worst case success probability is $P_{suc}^W=\frac{1}{2}$.\\
{\bf Most General Mixed Strategy}: We want to consider the most general mixed strategy such that $P_{suc}^W>\frac{1}{2}$. A generic such strategy can be realized in following steps.
\begin{itemize}
\item[(S-1)] If the x index is received, Alice tosses a $2$-sided biased coin having the outcomes $\{0,1\}$. The outcome probabilities of the coin are given by $P_x(0)=\alpha_x$ and $P_x(1)=1-\alpha_x$.
\item[(S-2)] Alice communicates the outcome of the coin toss to Bob through the perfect $1$-bit classical channel. 
\item[(S-3)] Bob prepares two $3$-sided coins with outcome probabilities specified by the probability vectors $\vec{r}=(r_1,r_2,r_3)$ and $\vec{q}=(q_1,q_2,q_3)$, respectively. If he receives $0$ from Alice he tosses the $\vec{r}$ coin and answers b if $b^{th}$ outcome occurs. He follows a similar strategy with the $\vec{q}$ coin if $1$ is received from Alice.
\end{itemize}
So the success probability for different inputs are given by $\gamma_1=P(1|1)=\alpha_1r_1+(1-\alpha_1)q_1$, $\gamma_2=P(2|2)=\alpha_2r_2+(1-\alpha_2)q_2$, $\gamma_3=P(3|3)=\alpha_3r_3+(1-\alpha_3)q_3$.For the worst case success probability to be greater than $\frac{1}{2}$ we require
\begin{subequations}
    \begin{align}
        \alpha_1r_1+(1-\alpha_1)q_1>\frac{1}{2}&\implies r_1>\frac{1}{2} \lor q_1>\frac{1}{2}\\
        \alpha_2r_2+(1-\alpha_2)q_2>\frac{1}{2}&\implies r_2>\frac{1}{2} \lor q_2>\frac{1}{2}\\
        \alpha_3r_3+(1-\alpha_3)q_3>\frac{1}{2}&\implies r_3>\frac{1}{2} \lor q_3>\frac{1}{2}
    \end{align}
\end{subequations}
All the equations cannot be satisfied simultaneously so the worst case success probability without Shared correlation is $P_{suc}^W=\frac{1}{2}$.\\
{\bf Classically Correlated Strategy}: If shared randomness is allowed between Alice and Bob the worst case success probability can be increased. Alice divides the Restaurants into two disjoint partitions $X$ and $Y$, such that $ |X|=1$ and $|Y|=2$. Alice sends $0$ to Bob whenever the index in $X$ is received, otherwise she sends $1$. Bob visits the Restaurants within the set $Y$ ($X$) with uniform probability if he receives $0~(1)$ from Alice.  Consider the following three partitionings $X$ and $Y$.
\begin{itemize}
\item[C-$1$:] $X=\{1\}~\&~Y=\{2,3\}$,
\item[C-$2$:] $X=\{2\}~\&~Y=\{1,3\}$,
\item[C-$3$:] $X=\{3\}~\&~Y=\{1,2\}$.
\end{itemize}
These correspond to three different strategies leading to three different correlation $P(b|x)$ 
\footnotesize
\begin{align*}
\mathbb{V}_1= \left(\begin{array}{ccc}
     1  &0&0 \\
    0&   \frac{1}{2}&\frac{1}{2}\\
     0&   \frac{1}{2}&\frac{1}{2}\\
  \end{array}\right),
\mathbb{V}_2= \left(\begin{array}{ccc}
     \frac{1}{2}&0&\frac{1}{2}\\
     0&1&0 \\
     \frac{1}{2}&0&\frac{1}{2}\\
  \end{array}\right),
\mathbb{V}_3=\left(\begin{array}{ccc}
      \frac{1}{2}&\frac{1}{2}&0\\
      \frac{1}{2}&\frac{1}{2}&0\\
        0&0&1 \\
  \end{array}\right).
\end{align*}
\normalsize
Equal mixture of these three strategies, which requires Alice and Bob to share $\log3$-bit of shared randomness, yields the resulting correlation
\begin{align*}
\mathbb{V}=\left(\begin{array}{ccc}
2/3&1/6&1/6\\
1/6&2/3&1/6\\
1/6&1/6&2/3\\
\end{array}\right),
\end{align*}
So the worst case success probability with the assistance of shared randomness is $P_{suc}^W=\frac{2}{3}$\\
{\bf Quantum Strategy}: In case of Quantum strategy Alice can send a bit system to Bob but they are not allowed to share any correlation. Bob will perform a measurement on the quantum system and depending on the outcome of measurement Bob will answer. Alice chooses a symmetric set of pure states lying on an equilateral triangle from a great circle of the Bloch sphere. When the $x$ index is given, she sends the state $\rho_x=\ket{\psi_x}\bra{\psi_x}=\frac{1}{2}\left[\mathbb{I}+\hat{n}_x\cdot \vec{\sigma}\right]$, where  $\hat{n}_x$ is the Bloch vector $\left(\sin\frac{2\pi (x-1)}{3},0,\cos\frac{2\pi (x-1)}{3}\right)^{\mathrm{T}}\in\mathbb{R}^3$; $x\in\{1,2,3\}$ and $\mathrm{T}$ denotes transposition. For decoding, Bob performs the measurement $\mathcal{M}(3)\equiv\left\{\pi_x:=\frac{1}{3}\left[\mathbb{I}+\hat{n}_x\cdot \vec{\sigma}\right]\right\}_{x=1}^3$, and answer x if outcome corresponding to the effect $\pi_x$ is observed. With this strategy we have $p(b=x|x)=\tr[\rho_i\pi_i]=\frac{2}{3},~\forall~x$. So in case of quantum strategy without any preshared correlation the worstcase success probability is given by $P_{suc}^W=\frac{2}{3}$.\\
This simpler game also demonstrate the utility of classical system can be enhanced with finite amount of Shared Randomness.

\section{Conclusions}
Our work brings to the forefront the efficacy of classical shared randomness in empowering the utility of classical communication. Although there has been a lot of research towards exploring the advantages of common-past resources to empower direct-communication resources, most of the investigations have been directed towards studying non-classical resources \cite{Bennett92,Thapliyal99,Frenkel21}. At this point it should be noted that nontrivial utility of classical shared randomness has also been studied in {\it Reverse Shannon Theorem}, both in classical and in quantum scenarios \cite{Winter02,Bennett02,Cubitt11,Bennett14}. In particular, the Ref.\cite{Cubitt11} is worth mentioning, where zero-error communication is studied using one classical channel to simulate another exactly in the presence of various classes of non-signalling correlations between sender and receiver {\it i.e.}, shared randomness, shared entanglement and arbitrary non-signalling correlations. It has been shown that while presence of classical shared randomness can be advantageous over the no correlation scenario, entanglement resource can provide further improvement over shared randomness. This in tern puts hierarchy  among different kinds of additional resources in the form of shared correlations. Important to note that no limitation on the amount of shared correlation is considered in this analysis. Our results are one step forward in this respect as it consider only finite amount of shared randomness which leads to comparison among different amount of same type of shared correlation. For instance, perfect winning of our strict game is not possible with 1 cbit+ 1 bit SR, while $log3$ bit SR along with 1 cbit suffice the purpose.   

The advantage reported here resembles the `quantum superdense coding' phenomenon, but with classical resources only. The crucial distinction is that in quantum superdense coding channel utility is quantified through mutual information, whereas in our case it is quantified through the success probability of winning the Restaurant game introduced here. The advantage is more striking than the one recently reported in \cite{Frenkel21}. There it has been shown that utility of a perfect classical channel can be empowered in assistant with preshared quantum entanglement, a non-classical common-past resource.  

Our proposed games are also important to reveal quantum advantage in the simplest communication scenario. While perfect winning of some games require shared randomness along with $1$-bit of classical communication, in quantum scenario they can be won with $1$-qubit communication alone. Importantly the advantage persists under experimental noises and hence welcomes novel experiments to implement the quantum protocol with presently available quantum technologies. The elimination aspect in our game (not visiting a closed Restaurant) seems to play a crucial role for the quantum advantage. One of the vastly acknowledged recent results in quantum foundation by Pusey-Barrett-Rudolph is worth mentioning here \cite{Pusey12}. This also invokes the power of elimination through entanglement measurement to establish `$\psi$-ontic' behaviour of quantum wave function. Our Restaurant games along with their variant turn out to be useful to put nontrivial `order of merit' among different combinations of communication resources as tabled in Section \ref{OoM}. In fact this observation motivates further research to explore many other such nontrivial ordering relations. An optimistic aim could be to come up with a game which can be won with $1$-qubit communication, whereas in classical domain it requires infinite amount of shared randomness along with $1$-bit of classical communication. The tasks studied in \cite{Galvao03,Perry15} might be the guiding lines at this point.   

The other important aspects are the no-go results established in Section \ref{sec:quantum_no_go}, which prove that an advantage of qubit over classical bit must invoke quantum superposition both at the encoding and decoding steps. At this point it might be interesting to recall an important result of quantum foundation recently established in \cite{Spekkens14}. While Bell in his seminal paper \cite{Bell64} and then Kochen \& Specker \cite{Kochen67} (see also Mermin \cite{Mermin93,Harrigan10}) have shown that for the projective measurements on a two level quantum system it is possible to come up with a `deterministic hidden variable model', in \cite{Spekkens14} it has been shown that such a model is not possible when considering more general POVM measurement. It is worth interesting to study whether the quantum advantage observed here has some deeper connection with this no-go result. The quantum advantages studied here also welcome a re-look to the claim made in \cite{Catani21}. The toy-bit model proposed there, originally motivated from \cite{Spekkens07}, seems incapable of showing the advantages reported in the present work.  

Our proposed game also opens up the opportunity to certify properties of different communication resources. Alice and Bob can certify the presence of direct communication resources by considering the $\mathbb{H}^n(1/n)$ game, as the input of Alice is completely random to Bob without any direct communication. If the direct communication resource is restricted to 1 cbit they will not be able to win $\mathbb{H}^n(1/n)$ game for any odd n without any preshared correlation. This can be useful in a practical scenario to verify the existence of preshared correlation between Alice and Bob by allowing them only 1 cbit direct communication.  \\

Many other interesting questions follow immediately. For instance, it is natural to ask whether shared randomness can empower other direct-Communication resources other than the classical channel. Is it possible to come with a communication scenario where the utility establishes some trade-off between classical communication and shared randomness. Lastly, our work highlights the fact that whether or not a Common-past resource empowers a direct-communication resource depends on the scenario under consideration, \textit{i.e.,} it depends on which of the inputs/outputs are trivial. A lot more interesting situations arise in the multipartite scenario since there are many more possibilities in which the inputs/outputs of a subset of parties can be made trivial. 
\appendix
\section{Quantum \& post-quantum enhancement of classical channel}\label{appendix-a}
In a recent work \cite{Frenkel21}, the authors have reported very interesting results. They have show that preshared entanglement can empower communication utility of a perfect classical channel. In fact post quantum correlation can do better than quantum correlations. They have shown it through the success probability of a two party game. Here, for the sake of completeness, we  review their game and show that the success probability scales with the well known CHSH expression. 

{\bf $4$-cup \& $2$-ball game:} Assume that there are $4$ cups and a referee randomly chooses two of them and put $1$ ball in each of them. Alice knows which two cups contain the balls but this information is not known to Bob. Bob needs to pick a cup with ball. If Alice is allowed to communicate $1$ C-bit to Bob then maximum success probability of Bob's winning turns out to be $5/6$ \cite{Frenkel21}. In the following we analyze the success probability when $1$ C-bit communication is assisted with no signaling (NS) correlation.    

{\bf NS correlation:} A two-party two-input two-output NS probability distribution can be represented as the following correlation table:
\footnotesize
\begin{align}
\mathbb{P}^{xy}_{ab}&\equiv\begin{blockarray}{ccccc}
xy/ab& 00 & 01 & 10 & 11 \\
\begin{block}{c(cccc)}
00 & p(00|00) & p(01|00) & p(10|00) & p(11|00) \\
01 & p(00|01) & p(01|01) & p(10|01) & p(11|01) \\
10 & p(00|10) & p(01|10) & p(10|10) & p(11|10) \\
11 & p(00|11) & p(01|11) & p(10|11) & p(11|11) \\
\end{block}
\end{blockarray}\\
&\equiv
\begin{pmatrix}
c_{00} & m_0-c_{00} & n_0-c_{00} & 1-m_0-n_0+c_{00} \\
c_{01} & m_0-c_{01} & n_1-c_{01} & 1-m_0-n_1+c_{01} \\
c_{10} & m_1-c_{10} & n_0-c_{10} & 1-m_1-n_0+c_{10} \\
c_{11} & m_1-c_{11} & n_1-c_{11} & 1-m_1-n_1+c_{11} \\
\end{pmatrix},
\end{align}
\normalsize
where $\max\{0,m_i+n+j-1\}\le c_{ij}\le\min\{m_i,n_j\}~\forall~i,j\in\{0,1\}$. 

{\bf Strategy:} Given $1$ bit of classical communication assisted with the aforesaid NS correlation, Alice will go by the following strategies depending on which of the cups combination is realized.  
\begin{table}[H]
\setlength
\extrarowheight{5pt}
\centering
\begin{tabular}{c| c} 
\hline
Cups & Alice's action\\
\hline\hline
$12$ & sends $0$ to Bob\\
\hline
$34$ & sends $1$ to Bob\\
\hline
$13$ & input $x=0$ and sends outcome $a$ to Bob\\
\hline
$24$ & input $x=0$ and sends outcome $a\oplus1$ to Bob\\
\hline
$14$ & input $x=1$ and sends outcome $a$ to Bob\\
\hline
$23$ &input $x=1$ and sends outcome $a\oplus1$ to Bob\\
\hline\hline
\end{tabular}
\end{table}
Now, depending on the bit received from Alice, Bob chooses his input for the NS box and depending on the outcome received from it he chooses the cup as follows:
\begin{center}
\begin{tabular}{c| c|c} 
\hline
Bit received & Bob's Action & Choice's of the cup\\
\hline\hline
\multirow{2}{3em}{~~~~$0$}& \multirow{2}{6em}{input $y=0$} & chooses cup $1$ for $b=0$\\
&  & chooses cup $2$ for $b=1$\\
\hline
\multirow{2}{3em}{~~~~$1$}& \multirow{2}{6em}{input $y=1$} & chooses cup $3$ for $b=0$\\
&  & chooses cup $4$ for $b=1$\\
\hline\hline
\end{tabular}
\end{center}
For the above strategy, the for pair of cups combinations the success probabilities become 
\begin{align*}
P_{succ}^{12}&=p(b=0|y=0)+p(b=1|y=0)=1,\\
P_{succ}^{34}&=p(b=0|y=1)+p(b=1|y=1)=1,\\
P_{succ}^{13}&=p(00|00)+p(10|01),
P_{succ}^{24}=p(11|00)+p(01|01),\\
P_{succ}^{14}&=p(00|10)+p(11|11),
P_{succ}^{23}=p(11|10)+p(00|11).
\end{align*}
Therefore, we have average success,
\begin{align}
P_{succ}&=\frac{1}{6}\left(P_{succ}^{12}+P_{succ}^{34}+P_{succ}^{13}+P_{succ}^{24}+P_{succ}^{13}+P_{succ}^{23}\right)\nonumber\\
&=\frac{1}{6}[5+2(c_{00}-c_{01}+c_{10}+c_{11}-m_1-n_0)].
\end{align}
Again, for the NS correlation the CHSH expression reads as
\begin{align}
\mathbb{CHSH}&=\langle x_0y_0\rangle-\langle x_0y_1\rangle+\langle x_1y_0\rangle+\langle x_1y_1\rangle\nonumber\\
&=2+4(c_{00}-c_{01}+c_{10}+c_{11}-m_1-n_0).
\end{align}
Here, $\langle x_iy_j\rangle:=\sum_{a,b=\pm}abp(ab|x_iy_j)$.
Thus, we have
\begin{align}
P_{succ}&=\frac{1}{12}(8+\mathbb{CHSH}). 
\end{align}
Clearly this success value is strictly greater than $1$-bit success value (i.e. $5/6$) wherever $\mathbb{CHSH}>2$. In other words, all NS correlation that are CHSH nonlocal (i.e. violate CHSH inequality) are advantageous in $4$-cup \& $2$-ball game and moreover post-quantum correlation yield better than quantum success.

\section{Playing $\mathbb{H}^3(\gamma_1,\gamma_2,\gamma_2)$ with $\mathcal{P}_{ly}(4)$}\label{gen-3-Restaurant with square}
For this generic case, Alice applies the encoding
\begin{align*}
\mathbb{E}^{(12)}\equiv\begin{cases}
1\mapsto\omega_1,~~2\mapsto\omega_4,\\
3\mapsto s:=q\times\omega_2 +(1-q)\times\omega_3.
\end{cases}
\end{align*}
Superscript in $\mathbb{E}^{(12)}$ denotes the fact that Restaurant $1$ and $2$ are encoded in extreme states, while some mixed state is used for the Restaurant $3$. Bob performs the measurement $\mathcal{M}\equiv\left\{p e_1,~(1-p) e_2,~p e_3,~(1-p) e_4\right\}$, and visits the Restaurants as follows 
\begin{figure}[t!]
\centering
\includegraphics[width=0.45\textwidth]{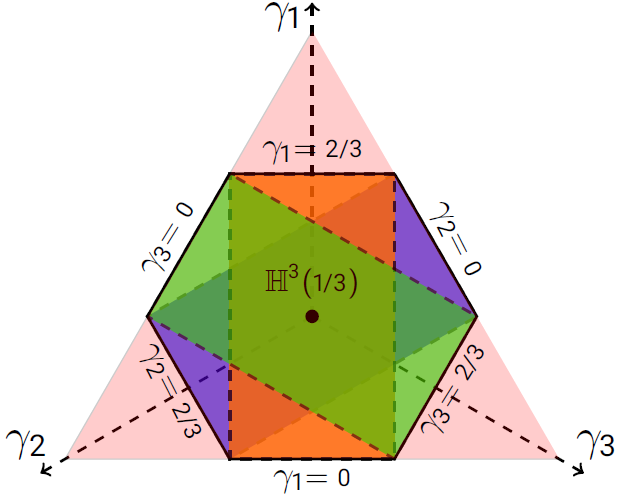}
\caption{Green rectangle consists of the games $\mathbb{H}^3(\gamma_1,\gamma_2,\gamma_3)$ that are winnable with encoding  $\mathbb{E}^{(12)}$. Violet and orange are the games winnable with encodings $\mathbb{E}^{(13)}$ and $\mathbb{E}^{(23)}$, respectively.}\label{fig13}
\end{figure}
\begin{align*}
\mathbb{D}^{(12)}\equiv\begin{cases}
p e_1\mapsto 3,\\
(1-p)e_2\mapsto2,\\
(1-p) e_4\mapsto 1,\\
p e_3\mapsto\begin{cases}
1~\mbox{with probability}~ r,\\
2~\mbox{with probability}~ (1-r).
\end{cases}
\end{cases}
\end{align*}  
Here, the superscript in $\mathbb{D}^{(12)}$ is used to indicate that this particular decoding is associated with the encoding $\mathbb{E}^{(12)}$. A straightforward calculation yields
\begin{align*}
p(1|2)&=p(2|1)=1-p,~~p(3|1)= p(3|2)=p,\\
p(1|3)&=p\times r+(1-p)\times (1-q),\\
p(2|3)&=(1-p)\times q+p\times (1-r).
\end{align*}
These subsequently gives us $\gamma_1=\frac{1}{3}[(1-p)+p\times r+(1-p)\times (1-q)]$,  $\gamma_2=\frac{1}{3}[(1-p)+(1-p)\times q+p\times (1-r)]$, and  $\gamma_3=\frac{2p}{3}$. Using the encoding strategy $\mathbb{E}^{(12)}$ with varying $p,q,r\in[0,1]$, the class of games highlighted in the green rectangle in Fig.\ref{fig13} can be won perfectly, which is a proper subset of all games. However, if we consider two other encoding-decoding strategies $\left(\mathbb{E}^{(13)},\mathbb{D}^{(13)}\right)$ and $\left(\mathbb{E}^{(23)},\mathbb{D}^{(23)}\right)$ with
\begin{align*}
\mathbb{E}^{(13)}&\equiv\begin{cases}
1\mapsto\omega_1,~~
3\mapsto\omega_4,\\
2\mapsto s:=q\times\omega_2 +(1-q)\times\omega_3.
\end{cases}\\
\mathbb{D}^{(13)}&\equiv\begin{cases}
p e_1\mapsto 2,\\
(1-p)e_2\mapsto 3,\\
(1-p) e_4\mapsto 1,\\
p e_3\mapsto\begin{cases}
1~\mbox{with probability}~ r,\\
3~\mbox{with probability}~ (1-r).
\end{cases}
\end{cases}
\end{align*}
\begin{align*}
\mathbb{E}^{(23)}&\equiv\begin{cases}
2\mapsto\omega_1,~~3\mapsto\omega_4,\\
1\mapsto s:=q\times\omega_2 +(1-q)\times\omega_3.
\end{cases}\\
\mathbb{D}^{(23)}&\equiv\begin{cases}
p e_1\mapsto 1,\\
(1-p)e_2\mapsto 3,\\
(1-p) e_4\mapsto 2,\\
p e_3\mapsto\begin{cases}
2~\mbox{with probability}~ r,\\
3~\mbox{with probability}~ (1-r).
\end{cases}
\end{cases}
\end{align*}
then it turns out that all the games $\mathbb{H}^3(\gamma_1,\gamma_2,\gamma_3)$ are perfectly winnable with communication of $\mathcal{P}_{ly}(4)$ system without requiring any shared randomness. 
\section{Classical correlated strategy for $\mathbb{H}^n[1/(n-1)]$}\label{ccsfor strict}
$\mathbf{\mathbb{H}^4[1/3]}:$ Although the game $\mathbb{H}^4[1/3]$ cannot be won with $1$-bit+$1$-SR resource, here we will show that the goal can be perfectly achieved through a $1$-bit classical channel if more shared randomness are provided as assistance. The strategy is described below.  

Alice divides the Restaurants into two disjoint partitions $X$ and $Y$, such that $ |X|=|Y|=2$. Alice sends $0$ to Bob whenever a Restaurant in $X$ is closed, otherwise she sends $1$. Bob visits the Restaurants within the set $Y$ ($X$) with uniform probability if he receives $0~(1)$ from Alice. Clearly,the condition ($\mathrm{h}_s1$) is satisfied as bob never visits a closed Restaurant. Consider the following three partitionings $X$ and $Y$.
\begin{itemize}
\item[C-$1$:] $X=\{1,2\}~\&~Y=\{3,4\}$,
\item[C-$2$:] $X=\{1,3\}~\&~Y=\{2,4\}$,
\item[C-$3$:] $X=\{1,4\}~\&~Y=\{2,3\}$.
\end{itemize}
These correspond to three different strategies leading to three different visit matrices 
\footnotesize
\begin{align*}
\mathbb{V}_1= \left(\begin{array}{cccc}
     0  &0&\frac{1}{2}&\frac{1}{2} \\
    0&0&    \frac{1}{2}&\frac{1}{2}\\
      \frac{1}{2}&\frac{1}{2}&0&0\\
         \frac{1}{2}&\frac{1}{2} & 0&0
  \end{array}\right),
\mathbb{V}_2= \left(\begin{array}{cccc}
      0&\frac{1}{2}&0&\frac{1}{2} \\
       \frac{1}{2}&0&\frac{1}{2}&0\\
     0&\frac{1}{2}&0&\frac{1}{2}\\
         \frac{1}{2}&0&\frac{1}{2} & 0
  \end{array}\right),
\mathbb{V}_3=\left(\begin{array}{cccc}
     0  & \frac{1}{2}&\frac{1}{2}&0 \\
        \frac{1}{2}&0&0&\frac{1}{2}\\
     \frac{1}{2}&0&0&\frac{1}{2}\\
         0&\frac{1}{2}&\frac{1}{2} & 0
  \end{array}\right).
\end{align*}
\normalsize
Equal mixture of these three strategies, which requires Alice and Bob to share $\log3$-bit of shared randomness, yields the resulting visit matrix 
\begin{align*}
\mathbb{V}=\left(\begin{array}{cccc}
0&1/3&1/3&1/3\\
1/3&0&1/3&1/3\\
1/3&1/3&0&1/3\\
1/3&1/3&1/3&0
\end{array}\right),
\end{align*}
which satisfies both the conditions ($\mathrm{h}_s1$) and ($\mathrm{h}_s2$).

$\mathbf{\mathbb{H}^n[1/(n-1)]}:$ The above protocol can be generalized for  $\mathbb{H}^{n}[1/(n-1)]$ game with $1$-bit of classical communication assisted with $\log (n-1)$-bit of SR. Alice sends $0$ and $1$ to direct Bob to visit $k$-th and $(k+1)$-th Restaurant respectively, where $k\in\{1,2,\cdots,(n-1)\}$. The value of the $k$ will be identified by the outcomes $\{1,2,\cdots,(n-1)\}$ of the SR, so it requires $\log (n-1)$-bit of SR. Whenever the $m$-th Restaurant ($m\in\{2,3,\cdots,(n-1)\}$) is closed, Alice communicates $0$ and $1$ respectively for every $k\in\{1,\cdots,(m-1)\}$-th and $k\in\{m,\cdots,(n-1)\}$-th outcomes of the SR. On the other hand, if the $1$-st or, $n$-th Restaurant is closed, then Alice will communicate $0$ and $1$ respectively, independent of the SR outcomes. It is easy to verify that the strategy satisfies both the required conditions. 

At this point it is important to see that a $\mathbb{H}^{n}[1/(n-1)]$ game can not be won with a qubit communication alone, whenever $n\ge5$. This establishes the `order of merit' Q$\prec_{inst}$C+SR as listed in Table \ref{table5}.

{\bf ACKNOWLEDGMENTS:} SGN acknowledges support from the CSIR project-\\ 09/0575(15951)/2022-EMR-I. EPL acknowledges support from the FWO through the BeQuNet SBO project
S008323N. TG is supported by the Hong Kong Research Grant
Council through Grants No. 17307719 and 17307520 and
though the Senior Research Fellowship Scheme SRFS2021-
7S02, by the Croucher Foundation, and by the John Templeton
Foundation through Grant No. 62312, The Quantum Information Structure of Space-time (qiss.fr). MA and MB acknowledge funding from the National Mission in Interdisciplinary Cyber-Physical systems from the Department of Science and Technology through the I-HUB Quantum Technology Foundation (Grant no: I-HUB/PDF/2021-22/008). MB acknowledges support through the research grant of INSPIRE Faculty fellowship from the Department of Science and Technology, Government of India, and the start-up research grant from SERB, Department of Science and Technology (Grant no: SRG/2021/000267).


\begin{thebibliography}{99}
\bibitem{Shannon48} C. E. Shannon; A mathematical theory of communication, \href{https://doi.org/10.1002/j.1538-7305.1948.tb01338.x}{Bell Syst. Tech. J. {\bf 27}, 379 (1948)}. 

\bibitem{Nielsen10}  M. A. Nielsen and I. L. Chuang; Quantum Computation and Quantum Information (Cambridge University Press, Cambridge, England, 2010).

\bibitem{Dowling03} J. P. Dowling and G. J. Milburn; Quantum technology: the second quantum revolution, 
\href{https://doi.org/10.1098/rsta.2003.1227}{Phil. Trans. R. Soc. Lond. A {\bf 361}, 1655 (2003)}


\bibitem{Bennett92} C. H. Bennett and S. J. Wiesner; Communication via one- and two-particle operators on Einstein-Podolsky-Rosen states,
\href{https://doi.org/10.1103/PhysRevLett.69.2881}{Phys. Rev. Lett. {\bf 69}, 2881 (1992)}.


\bibitem{Bennett93} C. H. Bennett, G. Brassard, C. Crépeau, R. Jozsa, A. Peres, and W. K. Wootters; Teleporting an unknown quantum state via dual classical and Einstein-Podolsky-Rosen channels, 
\href{https://doi.org/10.1103/PhysRevLett.70.1895}{Phys. Rev. Lett. {\bf 70}, 1895 (1993)}.

\bibitem{Bennett00} C. H. Bennett and D. DiVincenzo; Quantum information and computation,
\href{https://doi.org/10.1038/35005001}{Nature {\bf 404}, 247 (2000)}.

\bibitem{Kimble08} H. J. Kimble;  The quantum internet,
\href{https://doi.org/10.1038/nature07127}{Nature {\bf 453}, 1023 (2008)}. 

\bibitem{Dale15} H. Dale, D. Jennings, and T. Rudolph; Provable quantum advantage in randomness processing, 
\href{https://doi.org/10.1038/ncomms9203}{Nat. Commun. 6, 8203 (2015)}.

\bibitem{Zhang17} W. Zhang, D-S Ding, Y-B Sheng, L. Zhou, B-S Shi, and G-C Guo; Quantum Secure Direct Communication with Quantum Memory,
\href{https://doi.org/10.1103/PhysRevLett.118.220501}{Phys. Rev. Lett. {\bf 118}, 220501 (2017)}.

\bibitem{Boes18} P. Boes, H. Wilming, R. Gallego, and J. Eisert; Catalytic Quantum Randomness,
\href{https://doi.org/10.1103/PhysRevX.8.041016}{Phys. Rev. X {\bf 8}, 041016 (2018)}.

\bibitem{Rosset18} D. Rosset, F. Buscemi, and Y-C. Liang; Resource Theory of Quantum Memories and Their Faithful Verification with Minimal Assumptions, 
\href{https://doi.org/10.1103/PhysRevX.8.021033}{Phys. Rev. X {\bf 8}, 021033 (2018)}.

\bibitem{Ebler18} D. Ebler, S. Salek, and G. Chiribella; Enhanced Communication with the Assistance of Indefinite Causal Order,
\href{https://doi.org/10.1103/PhysRevLett.120.120502}{Phys. Rev. Lett. {\bf 120}, 120502 (2018)}.

\bibitem{Korzekwa21} K. Korzekwa and M. Lostaglio; Quantum Advantage in Simulating Stochastic Processes,
\href{https://doi.org/10.1103/PhysRevX.11.021019}{Phys. Rev. X {\bf 11}, 021019 (2021)}.

\bibitem{Chiribella21} G. Chiribella, M. Banik, S. S. Bhattacharya, T. Guha, M. Alimuddin, A. Roy, S. Saha, S. Agrawal, and G. Kar; Indefinite causal order enables perfect quantum communication with zero capacity channels, 
\href{https://doi.org/10.1088/1367-2630/abe7a0}{New J. Phys. {\bf 23}, 033039 (2021)}.

\bibitem{Bhattacharya21} S. S. Bhattacharya, A. G. Maity, T. Guha, G. Chiribella, and M. Banik; Random-Receiver Quantum Communication, 
\href{https://doi.org/10.1103/PRXQuantum.2.020350}{PRX Quantum {\bf 2}, 020350 (2021)}.

\bibitem{Koudia21} S. Koudia, A. S. Cacciapuoti, and M. Caleffi; How Deep the Theory of Quantum Communications Goes: Superadditivity, Superactivation and Causal Activation, 
\href{https://ieeexplore.ieee.org/document/9851437}{IEEE Commun. Surv. Tutor. {\bf 24} (4), 1926-1956 (2022)}.


\bibitem{Bouwmeester97} D. Bouwmeester, J. W. Pan, K. Mattle, M. Eibl, H. Weinfurter, and A. Zeilinger; Experimental quantum teleportation,
\href{https://doi.org/10.1038/37539}{Nature {\bf 390}, 575 (1997)}. 

\bibitem{Gisin02} N. Gisin, G. Ribordy, W. Tittel, and H. Zbinden; Quantum Cryptography, 
\href{https://doi.org/10.1103/RevModPhys.74.145}{Rev. Mod. Phys. {\bf 74}, 145 (2002)}. 

\bibitem{Georgescu14}  I. M. Georgescu, S. Ashhab, and F. Nori; Quantum Simulation, 
\href{https://doi.org/10.1103/RevModPhys.86.153}{Rev. Mod. Phys. {\bf  86}, 153 (2014)}.

\bibitem{Degen17} C. L. Degen, F. Reinhard, and P. Cappellaro; Quantum Sensing, 
\href{https://doi.org/10.1103/RevModPhys.89.035002}{Rev. Mod. Phys. {\bf 89}, 035002 (2017)}.

\bibitem{Yin17} J. Yin {\it et al.} Satellite-based entanglement distribution over 1200 kilometers,
\href{https://doi.org/10.1126/science.aan3211}{Science {\bf 356}, 1140 (2017)}.

\bibitem{Valivarthi20} R. Valivarthi {\it et al.} Teleportation Systems Toward a Quantum Internet, 
\href{https://doi.org/10.1103/PRXQuantum.1.020317}{PRX Quantum {\bf 1}, 020317 (2020)}.

\bibitem{Xu20} F. Xu, X. Ma, Q. Zhang, H-K Lo, and J-W Pan; Secure quantum key distribution with realistic devices, 
\href{https://doi.org/10.1103/RevModPhys.92.025002}{Rev. Mod. Phys. {\bf 92}, 025002 (2020)}. 

\bibitem{Holevo73} A. S. Holevo; Bounds for the Quantity of Information Transmitted by a Quantum Communication Channel,
\href{http://www.mathnet.ru/php/archive.phtml?wshow=paper&jrnid=ppi&paperid=903&option_lang=eng}{Problems Inform. Transmission {\bf 9}, 177 (1973)}.

\bibitem{Mermin04} N. D. Mermin; Copenhagen computation: How I learned to stop worrying and love Bohr,
\href{https://doi.org/10.1147/rd.481.0053}{IBM J. Res. Dev. {\bf 48}, 53 (2004)}.

\bibitem{Frenkel15} P. E. Frenkel and M. Weiner; Classical information storage in an $n$-level quantum system, 
\href{https://doi.org/10.1007/s00220-015-2463-0}{Comm. Math. Phys. {\bf 340}, 563 (2015)}.

\bibitem{Bell64} J.S. Bell; On the Einstein Podolsky Rosen paradox,
\href{https://doi.org/10.1103/PhysicsPhysiqueFizika.1.195}{Physics {\bf 1}, 195 (1964)}.	

\bibitem{Bell66} J. S. Bell; On the Problem of Hidden Variables in Quantum Mechanics,
\href{https://doi.org/10.1103/RevModPhys.38.447}{Rev. Mod. Phys. {\bf 38}, 447 (1966)}.

\bibitem{Brunner14(0)} N. Brunner, D. Cavalcanti, S. Pironio, V. Scarani, and S. Wehner; Bell nonlocality,
\href{https://doi.org/10.1103/RevModPhys.86.419}{Rev. Mod. Phys. {\bf 86}, 419 (2014)}.

\bibitem{Wolfe20} E. Wolfe, D. Schmid, A. B. Sainz, R. Kunjwal, and R. W. Spekkens; Quantifying Bell: the Resource Theory of Nonclassicality of Common-Cause Boxes, 
\href{https://doi.org/10.22331/q-2020-06-08-280}{Quantum {\bf 4}, 280 (2020)}.

\bibitem{Schmid20} D. Schmid, D. Rosset, and F. Buscemi; The type-independent resource theory of local operations and shared randomness, \href{https://doi.org/10.22331/q-2020-04-30-262}{Quantum {\bf 4}, 262 (2020)}.

\bibitem{Rosset20} D. Rosset, D. Schmid, and F. Buscemi; Type-Independent Characterization of Spacelike Separated Resources, 
\href{https://doi.org/10.1103/PhysRevLett.125.210402}{Phys. Rev. Lett. {\bf 125}, 210402 (2020)}.

\bibitem{Aumann87} R. J. Aumann; Correlated equilibrium as an expression of bayesian rationality, \href{https://doi.org/10.2307/1911154}{Econometrica {\bf 55}, 1 (1987)}. 

\bibitem{Babai97} L. Babai and P. G. Kimmel; Randomized simultaneous messages: solution of a problem of Yao in communication complexity; \href{https://doi.org/10.1109/ccc.1997.612319}{Proc. Compu. Complexity. 20th Annual IEEE Conference (1997)}.

\bibitem{Canonne17} C. L. Canonne, V. Guruswami, R. Meka, and M. Sudan; Communication with imperfectly shared randomness, 
\href{https://doi.org/10.1109/tit.2017.2734103}{IEEE Trans. Inf. Theory {\bf 63}, 6799 (2017)}. 

\bibitem{Toner03} B. F. Toner and D. Bacon; Communication cost of simulating bell correlations,
\href{https://doi.org/10.1103/PhysRevLett.91.187904}{Phys. Rev. Lett. {\bf 91}, 187904 (2003)}.

\bibitem{Bowles15} J. Bowles, F. Hirsch, M. T. Quintino, and N. Brunner; Local hidden variable models
for entangled quantum states using finite shared randomness,
\href{https://doi.org/10.1103/PhysRevLett.114.120401}{Phys. Rev. Lett. {\bf 114}, 120401 (2015)}.

\bibitem{Llobet15} M. Perarnau-Llobet, K. V. Hovhannisyan, M. Huber, P. Skrzypczyk, N. Brunner, and A. Acín; Extractable Work from Correlations,
\href{https://doi.org/10.1103/PhysRevX.5.041011}{Phys. Rev. X {\bf 5}, 041011 (2015)}.

\bibitem{Guha21} T. Guha, M. Alimuddin, S. Rout, A. Mukherjee, S. S. Bhattacharya, and M. Banik; Quantum Advantage for Shared Randomness Generation, 
\href{https://doi.org/10.22331/q-2021-10-27-569}{Quantum {\bf 5}, 569 (2021)}.

\bibitem{Janotta11} P. Janotta, C. Gogolin, J. Barrett, and N. Brunner; Limits on nonlocal correlations from the structure of the local state space, \href{https://doi.org/10.1088/1367-2630/13/6/063024}{New J. Phys. {\bf 13}, 063024 (2011)}.

\bibitem{Horodecki09} R. Horodecki, P. Horodecki, M. Horodecki, and K. Horodecki; Quantum entanglement, \href{https://doi.org/10.1103/RevModPhys.81.865}{Rev. Mod. Phys. {\bf 81}, 865 (2009)}.

\bibitem{Popescu94} S. Popescu and  D. Rohrlich ; Quantum nonlocality as an axiom, 
\href{https://doi.org/10.1007/BF02058098}{Found. Phys. {\bf 24}, 379 (1994)}.

\bibitem{Barrett07} J. Barrett; Information processing in generalized probabilistic theories, 
\href{https://doi.org/10.1103/PhysRevA.75.032304}{Phys. Rev. A {\bf 75}, 032304 (2007)}.

\bibitem{Brunner14} N. Brunner, M. Kaplan, A. Leverrier, and P. Skrzypczyk; Dimension of physical systems, information processing, and thermodynamics,
\href{https://doi.org/10.1088/1367-2630/16/12/123050}{New J. Phys. {\bf 16}, 123050 (2014)}.


\bibitem{Hall11} M. J. W. Hall; Relaxed Bell inequalities and Kochen-Specker theorems,
\href{https://doi.org/10.1103/PhysRevA.84.022102}{Phys. Rev. A {\bf 84}, 022102 (2011)}. 

\bibitem{Banik13} M. Banik; Lack of measurement independence can simulate quantum correlations even when signaling can not,
\href{https://doi.org/10.1103/PhysRevA.88.032118}{Phys. Rev. A {\bf 88}, 032118 (2013)}.

\bibitem{Schaetz04} T. Schaetz, M. D. Barrett, D. Leibfried, J. Chiaverini, J. Britton, W. M. Itano, J. D. Jost, C. Langer, and D. J. Wineland; Quantum Dense Coding with Atomic Qubits,
\href{https://doi.org/10.1103/PhysRevLett.93.040505}{Phys. Rev. Lett. {\bf 93}, 040505 (2004)}.

\bibitem{Barreiro08} J. Barreiro, T. C. Wei, and P. Kwiat; Beating the channel capacity limit for linear photonic superdense coding,
\href{https://doi.org/10.1038/nphys919}{Nature Phys {\bf 4}, 282 (2008)}. 

\bibitem{Williams17} B. P. Williams, R. J. Sadlier, and T. S. Humble; Superdense Coding over Optical Fiber Links with Complete Bell-State Measurements, 
\href{https://doi.org/10.1103/PhysRevLett.118.050501}{Phys. Rev. Lett. {\bf 118}, 050501 (2017)}.

\bibitem{Thapliyal99} C. H. Bennett, P. W. Shor, J. A. Smolin, and A. V. Thapliyal; Entanglement-Assisted Classical Capacity of Noisy Quantum Channels,
\href{https://doi.org/10.1103/PhysRevLett.83.3081}{Phys. Rev. Lett. {\bf 83}, 3081 (1999)}.

\bibitem{Frenkel21} P. E. Frenkel and M. Weiner; On entanglement assistance to a noiseless classical channel, 
\href{https://doi.org/10.22331/q-2022-03-01-662}{Quantum {\bf 6}, 662 (2022)}.

\bibitem{Clauser69} J. F. Clauser, M. A. Horne, A. Shimony, and R. A. Holt; Proposed Experiment to Test Local Hidden-Variable Theories,
\href{https://doi.org/10.1103/PhysRevLett.23.880}{Phys. Rev. Lett. {\bf 23}, 880 (1969)}.

\bibitem{DallArno17} M. Dall'Arno, S. Brandsen, A. Tosini, F. Buscemi, and V. Vedral; No-Hypersignaling Principle,
\href{https://doi.org/10.1103/PhysRevLett.119.020401}{Phys. Rev. Lett. {\bf 119}, 020401 (2017)}.

\bibitem{Wiesner83} S. Wiesner; Conjugate coding,
\href{https://doi.org/10.1145/1008908.1008920}{ACM Sigact News {\bf 15}, 78 (1983)}.

\bibitem{Ambainis99} A. Ambainis, A. Nayak, A. Ta-Shma, and U. Vazirani; Dense quantum coding and a lower bound for 1-way quantum automata,
\href{https://doi.org/10.1145/301250.301347}{in Proceedings of the thirty-first annual ACM symposium on
Theory of Computing (1999) pp. 376–383}.

\bibitem{Ambainis02} A. Ambainis, A. Nayak, A. Ta-Shma, and U. Vazirani; Dense quantum coding and quantum finite automata,
\href{https://doi.org/10.1145/581771.581773}{J. ACM {\bf 49}, 496 (2002)}.

\bibitem{Spekkens09} R. W. Spekkens, D. H. Buzacott, A. J. Keehn, B. Toner, G. J. Pryde; Preparation contextuality powers parity-oblivious multiplexing, 
\href{https://doi.org/10.1103/PhysRevLett.102.010401}{Phys. Rev. Lett. {\bf 102}, 010401 (2009)}.

\bibitem{Banik15} M. Banik, S. S. Bhattacharya, A. Mukherjee, A. Roy, A. Ambainis, A. Rai; Limited preparation contextuality in quantum theory and its relation to the Cirel'son bound, 
\href{https://doi.org/10.1103/PhysRevA.92.030103}{Phys. Rev. A {\bf 92}, 030103(R) (2015)}.

\bibitem{Czekaj17} L. Czekaj, M. Horodecki, P. Horodecki, and R. Horodecki; Information content of systems as a physical principle, 
\href{https://doi.org/10.1103/PhysRevA.95.022119}{Phys. Rev. A {\bf 95}, 022119 (2017)}.

\bibitem{Ambainis19} A. Ambainis, M. Banik, A. Chaturvedi, D. Kravchenko, and A. Rai; Parity oblivious d-level random access codes and class of noncontextuality inequalities, 
\href{https://doi.org/10.1007/s11128-019-2228-3}{Quantum Inf Process {\bf 18}, 111 (2019)}. 

\bibitem{Horodecki19} D. Saha, P. Horodecki, and M. Pawłowski; State independent contextuality advances one-way communication, 
\href{https://doi.org/10.1088/1367-2630/ab4149}{New J. Phys. {\bf 21}, 093057 (2019)}.

\bibitem{Saha19} D. Saha and A. Chaturvedi; Preparation contextuality as an essential feature underlying quantum communication advantage,
\href{https://doi.org/10.1103/PhysRevA.100.022108}{Phys. Rev. A {\bf 100}, 022108 (2019)}.

\bibitem{Vaisakh21} Vaisakh M, R. K. Patra, M. Janpandit, S. Sen, and M. Banik, and A. Chaturvedi; Mutually unbiased balanced functions and generalized random access codes, 
\href{https://doi.org/10.1103/PhysRevA.104.012420}{Phys. Rev. A {\bf 104}, 012420 (2021)}.

\bibitem{Naik21} S. G. Naik, E. P. Lobo, S. Sen, R. K. Patra, M. Alimuddin, T. Guha, S. S. Bhattacharya, and M. Banik; On composition of multipartite quantum systems: perspective from time-like paradigm, \href{https://doi.org/10.1103/PhysRevLett.128.140401}{Phys. Rev. Lett. {\bf 128}, 140401 (2022)}.


\bibitem{Ambainis08} A. Ambainis, D. Leung, L. Mancinska, and M. Ozols; Quantum Random Access Codes with Shared Randomness, 
\href{https://arxiv.org/abs/0810.2937}{arXiv:0810.2937 [quant-ph]}.

\bibitem{Pawlowski10} M. Pawłowski and M. Żukowski; Entanglement-assisted random access codes, 
\href{https://doi.org/10.1103/PhysRevA.81.042326}{Phys. Rev. A {\bf 81}, 042326 (2010)}.

\bibitem{Tavakoli21} A. Tavakoli, J. Pauwels, E. Woodhead, and S. Pironio; Correlations in Entanglement-Assisted Prepare-and-Measure Scenarios, 
\href{https://doi.org/10.1103/PRXQuantum.2.040357}{PRX Quantum {\bf 2}, 040357 (2021)}.

\bibitem{Bourennane} A. Piveteau, J. Pauwels, E. Håkansson, S. Muhammad, M. Bourennane, and A. Tavakoli; Entanglement-assisted quantum communication with simple measurements,
\href{https://doi.org/10.1038/s41467-022-33922-5}{Nat. Commun. {\bf 13}, 7878 (2022)}.

\bibitem{vanDam} W van Dam; Nonlocality \& Communication Complexity (PhD Thesis). 

\bibitem{Brassard06} G. Brassard, H. Buhrman, N. Linden, A. A. Méthot, A. Tapp, and F. Unger; Limit on Nonlocality in Any World in Which Communication Complexity Is Not Trivial,
\href{https://doi.org/10.1103/PhysRevLett.96.250401}{Phys. Rev. Lett. {\bf 96}, 250401 (2006)}.

\bibitem{Buhrman10} H. Buhrman, R. Cleve, S. Massar, and R. de Wolf; Nonlocality and communication complexity, 
\href{https://doi.org/10.1103/RevModPhys.82.665}{Rev. Mod. Phys. {\bf 82}, 665 (2010)}.

\bibitem{Mermin93} N. D. Mermin; Hidden variables and the two theorems of John Bell,
\href{https://doi.org/10.1103/RevModPhys.65.803}{Rev. Mod. Phys. {\bf 65}, 803 (1993)}.

\bibitem{Cirelson80} B. S. Cirel'son; Quantum generalizations of Bell's inequality,
\href{https://doi.org/10.1007/bf00417500}{Lett. Math. Phys. {\bf 4}, 93 (1980)}. 	

\bibitem{Slofstra20} W. Slofstra; Tsirelson's problem and an embedding theorem for groups arising from non-local games, 
\href{https://doi.org/10.1090/jams/929}{J. Amer. Math. Soc. {\bf 33}, 1 (2020)} (also \href{https://arxiv.org/abs/1606.03140}{	arXiv:1606.03140 [quant-ph]}).

\bibitem{Ji20} Z. Ji, A. Natarajan, T. Vidick, J. Wright, and H. Yuen; MIP*=RE,
\href{https://arxiv.org/abs/2001.04383}{arXiv:2001.04383 [quant-ph]}.

\bibitem{Fritz21} T. Fritz; Quantum logic is undecidable,
\href{https://doi.org/10.1007/s00153-020-00749-0}{Arch. Math. Logic {\bf 60}, 329 (2021)} (also \href{https://arxiv.org/abs/1607.05870}{arXiv:1607.05870 [quant-ph]}).

\bibitem{Buscemi12} F. Buscemi; All Entangled Quantum States Are Nonlocal,
\href{https://doi.org/10.1103/PhysRevLett.108.200401}{Phys. Rev. Lett. {\bf 108}, 200401 (2012)}.

\bibitem{Branciard13} C. Branciard, D. Rosset, Y-C Liang, and N. Gisin; Measurement-Device-Independent Entanglement Witnesses for All Entangled Quantum States, 
\href{https://doi.org/10.1103/PhysRevLett.110.060405}{Phys. Rev. Lett. {\bf 110}, 060405 (2013)}.

\bibitem{Lobo21} E. P. Lobo, S. G. Naik, S. Sen, R. K. Patra, M. Banik, and M. Alimuddin; Certifying beyond quantumness of locally quantum no-signaling theories through a quantum-input Bell test,
\href{https://doi.org/10.1103/PhysRevA.106.L040201}{Phys. Rev. A {\bf 106}, L040201 (2022)}.

\bibitem{Nash} J. F. Nash; Equilibrium points in n-person games,
\href{https://doi.org/10.1073/pnas.36.1.48}{PNAS {\bf 36}, 48 (1950)}; Non-cooperative games, 
\href{https://doi.org/10.2307/1969529}{Ann. Math. {\bf 54}, 286295 (1951)}.

\bibitem{Harsanyi} J. C. Harsanyi; Games with Incomplete Information Played by “Bayesian” Players, Part I. The Basic Model, 
\href{https://doi.org/10.1287/mnsc.14.3.159}{Management Science {\bf 14}, 159 (1967)};
Part II. Bayesian Equilibrium Points,
\href{https://doi.org/10.1287/mnsc.14.5.320}{Management Science {\bf 14}, 320 (1968)};
Part III. The Basic Probability Distribution of the Game,
\href{https://doi.org/10.1287/mnsc.14.7.486}{Management Science {\bf 14}, 486 (1968)}.

\bibitem{Papadimitriou08} C. H. Papadimitriou and T. Roughgarden; Computing correlated equilibria in multi-player
games,
\href{https://doi.org/10.1145/1379759.1379762}{J. ACM {\bf 55}, 14 (2008)}.

\bibitem{Brunner13} N. Brunner and N. Linden; Connection between Bell nonlocality and Bayesian game theory,
\href{https://doi.org/10.1038/ncomms3057}{Nat. Commun. {\bf 4}, 2057 (2013)}. 

\bibitem{Pappa15} A. Pappa, N. Kumar, T. Lawson, M. Santha, S. Zhang, E. Diamanti, and I. Kerenidis; Nonlocality and Conflicting Interest Games,
\href{https://doi.org/10.1103/PhysRevLett.114.020401}{Phys. Rev. Lett. {\bf 114}, 020401 (2015)}.

\bibitem{Roy16} A. Roy, A. Mukherjee, T. Guha, S. Ghosh, S. S. Bhattacharya, and M. Banik; Nonlocal correlations: Fair and unfair strategies in Bayesian games, 
\href{https://doi.org/10.1103/PhysRevLett.114.020401}{Phys. Rev. A {\bf 94}, 032120 (2016)}.

\bibitem{Banik19(1)} M. Banik, S. S. Bhattacharya, N. Ganguly, T. Guha, A. Mukherjee, A. Rai, and A. Roy; Two-Qubit Pure Entanglement as Optimal Social Welfare Resource in Bayesian Game,
\href{https://doi.org/10.22331/q-2019-09-09-185https://doi.org/10.22331/q-2019-09-09-185}{Quantum {\bf 3}, 185 (2019)}.


\bibitem{Mermin02} N. D. Mermin; Deconstructing dense coding,
\href{https://doi.org/10.1103/PhysRevA.66.032308}{Phys. Rev. A {\bf 66}, 032308 (2002)}.

\bibitem{Massar14} S. Massar and M. K. Patra; Information and communication in polygon theories, \href{https://doi.org/10.1103/PhysRevA.89.052124}{Phys. Rev. A {\bf 89}, 052124 (2014)}.

\bibitem{Muller12} M. P. Müller and C. Ududec; Structure of Reversible Computation Determines the Self-Duality of Quantum Theory,
\href{https://doi.org/10.1103/PhysRevLett.108.130401}{Phys. Rev. Lett. {\bf 108}, 130401 (2012)}. 

\bibitem{Safi15} S. W. Al-Safi and J. Richens; Reversibility and the structure of the local state space,
\href{https://doi.org/10.1088/1367-2630/17/12/123001}{New J. Phys. {\bf 17}, 123001 (2015)}.

\bibitem{Banik19} M. Banik, S. Saha, T. Guha, S. Agrawal, S. S. Bhattacharya, A. Roy, and A. S. Majumdar; Constraining the state space in any physical theory with the principle of information symmetry, 
\href{https://doi.org/10.1103/PhysRevA.100.060101}{Phys. Rev. A {\bf 100}, 060101(R) (2019)}.

\bibitem{Saha20} S. Saha, S. S. Bhattacharya, T. Guha, S. Halder, and M. Banik; Advantage of Quantum Theory over Nonclassical Models of Communication,
\href{https://doi.org/10.1002/andp.202000334}{Annalen der Physik \bf{532}, 2000334 (2020)}.

\bibitem{Bhattacharya20} S. S. Bhattacharya, S. Saha, T. Guha, and M. Banik; Nonlocality without entanglement: Quantum theory and beyond, 
\href{https://doi.org/10.1103/PhysRevResearch.2.012068}{Phys. Rev. Research {\bf 2}, 012068(R) (2020)}.

\bibitem{Winter02} A. Winter; Compression of sources of probability distributions and density operators,
\href{https://arxiv.org/abs/quant-ph/0208131}{arXiv:quant-ph/0208131}.

\bibitem{Bennett02} C. H. Bennett, P. W. Shor, J. A. Smolin, A. V. Thapliyal; Entanglement-assisted capacity of a quantum channel and the reverse Shannon theorem, 
\href{https://doi.org/10.1109/TIT.2002.802612}{IEEE Trans. Inf. Theory {\bf 48}, 2637 (2002)}.

\bibitem{Cubitt11} T. S. Cubitt, D. Leung, W. Matthews, A. Winter; Zero-error channel capacity and simulation assisted by non-local correlations,
\href{https://doi.org/10.1109/TIT.2011.2159047}{IEEE Trans. Info. Theory {\bf 57}, 5509 (2011)}.

\bibitem{Bennett14} C. H. Bennett, I. Devetak, A. W. Harrow, P. W. Shor, A.Winter; Quantum Reverse Shannon Theorem, 
\href{https://doi.org/10.1109/TIT.2014.2309968}{IEEE Trans. Inf. Theory {\bf 60}, 2926 (2014)}.

\bibitem{Pusey12} M. Pusey, J. Barrett, and T. Rudolph; On the reality of the quantum state,
\href{https://doi.org/10.1038/nphys2309}{Nat. Phys. {\bf 8}, 475 (2012)}. 

\bibitem{Galvao03} E. F. Galvão and L. Hardy; Substituting a Qubit for an Arbitrarily Large Number of Classical Bits, 
\href{https://doi.org/10.1103/PhysRevLett.90.087902}{Phys. Rev. Lett. {\bf 90}, 087902 (2003)}.

\bibitem{Perry15} C. Perry, R. Jain, and J. Oppenheim; Communication Tasks with Infinite Quantum-Classical Separation, \href{https://doi.org/10.1103/PhysRevLett.115.030504}{Phys. Rev. Lett. {\bf 115}, 030504 (2015)}.

\bibitem{Spekkens14} R. W. Spekkens; The Status of Determinism in Proofs of the Impossibility of a Noncontextual Model of Quantum Theory,
\href{https://doi.org/10.1007/s10701-014-9833-x}{Found. Phys. {\bf 44}, 1125 (2014)}.

\bibitem{Kochen67} S. Kochen and E. P. Specker; The problem of hidden variables in quantum mechanics,
\href{https://doi.org/10.1512/iumj.1968.17.17004}{J. Math. Mech. {\bf 17}, 59 (1967)}.

\bibitem{Harrigan10} N. Harrigan and R. W. Spekkens; Einstein, incompleteness, and the epistemic view of quantum states, 
\href{https://doi.org/10.1007/s10701-009-9347-0}{Found. Phys. {\bf 40}, 125 (2010)}.

\bibitem{Catani21} L. Catani, M. Leifer, D. Schmid, and R. W. Spekkens; Why interference phenomena do not capture the essence of quantum theory, \href{https://doi.org/10.22331/q-2023-09-25-1119}{Quantum {\bf 7}, 1119 (2023)}.

\bibitem{Spekkens07} R. W. Spekkens; Evidence for the epistemic view of quantum states: A toy theory,
\href{https://doi.org/10.1103/PhysRevA.75.032110}{Phys. Rev. A {\bf 75}, 032110 (2007)}.

\end{thebibliography}
\end{document}